\documentclass[onecolumn,conference]{article} 
\usepackage[margin=1in]{geometry}
\usepackage[hidelinks]{hyperref}

\usepackage[noend]{algpseudocode}
\usepackage{algorithm}
\usepackage{tikz}
\usetikzlibrary{fit}
\usepackage{fancyhdr}
\usepackage{mathtools}
\usepackage{xspace}
\usepackage{setspace}
\usepackage{ amssymb }
\usepackage{ amsfonts}
\usepackage{amsthm}
\usepackage{amsmath}
\usepackage{mathtools}
\usepackage{graphicx}
\usepackage{color,soul}
\usepackage{paralist}
\usepackage{enumitem}
\usepackage{varwidth}
\usepackage{cleveref}
\usepackage{paralist}
\usepackage{float}
\usepackage[justification=centering]{caption}

\newtheorem{theorem}{Theorem}
\newtheorem{corollary}{Corollary}
\newtheorem{lemma}{Lemma}
\newtheorem{definition}{Definition}
\newtheorem{claim}{Claim}
\newtheorem{assumption}{Assumption}
\usetikzlibrary{arrows,shadows,positioning,calc,math}
\usetikzlibrary{decorations.pathreplacing}

\algnewcommand\algicinput{\textbf{Input:}}
\algnewcommand\INPUT{\item[\algicinput]}
\algnewcommand\algicoutput{\textbf{Output:}}
\algnewcommand\OUTPUT{\item[\algicoutput]}

\usepackage{mdframed}
\usepackage{lipsum}

\newcommand{\alg}{Time Capsule\xspace}
\newcommand{\bname}{Capsule Chain\xspace}
\newcommand{\pcrdl}{PCR-DL\xspace}
\newcommand{\pcrdlalg}{PCRGenPuz\xspace}
\newcommand{\problem}{MPTC\xspace}
\newcommand{\problemlong}{Multi-Party Timed Commitments\xspace}

\newmdtheoremenv{defibox}{Definition}

\newcommand{\cmd}[1]{\ensuremath{#1}\xspace}
\newcommand{\gamemark}{\cmd{\Gamma}}

\newcommand{\gamee}[2]{\cmd{\gamemark_{\textit{#1}}^{\textit{#2}}}}
\newcommand{\Ad}{\cmd{\mathcal{A}}}

\newcommand\mysmaller{\stackrel{\mathclap{\normalfont\mbox{\normalfont\tiny$\lambda>c$}}}{<}}

\newcommand{\randselect}{\xleftarrow[\text{}]{\text{R}}}
\newcommand{\lambdam}{{\lambda_m}}
\newcommand{\lambdah}{{\lambda_h}}
\newcommand{\yboxbegin}{\noindent\fbox\bgroup\parbox\bgroup\textwidth\egroup\bgroup}
\newcommand{\yboxend}{\egroup\egroup}
		
\usepackage[colorinlistoftodos]{todonotes}
\newcounter{todocounter}

\newcommand\refc{c\ref}
\newcommand\refp{p\ref}
\newcommand\fun{\textsc}

\newcommand{\mytextbox}[1]{#1}

\makeatother
\date{}
\begin{document}
\title{{\problemlong}}

\author{ Yael Doweck \\Technion, IC3 \and  Ittay Eyal \\Technion, IC3}


%
%

\maketitle
\begin{abstract}
The problem of obtaining secret commitments from multiple parties and revealing them after a certain time is useful for sealed-bid auctions, games, and other applications. 
Existing solutions, dating back to Rivest, Shamir and Wagner, either do not scale or rely on synchrony for the commitment phase and trust of~$t/n$ parties. 

We formalize the problem of implementing such commitments with a probabilistic delay and without the aforementioned assumptions as~\emph{Multi-Party Timed Commitments~(MPTC)} and present a solution~-- the \emph{Time-Capsule} protocol. 
Like previous approaches, \alg forms a puzzle whose solution reveals the committed values. 
But unlike previous solutions, no party has an advantage in solving the puzzle, and individual commitments cannot be revealed before the entire set is committed. 

A particular application of \problem realizes an advancement in the study of decentralized systems. 
The state of the art in decentralized systems is manifested in blockchain systems that utilize Proof of Work to achieve censorship resistance. 
However, they are still vulnerable to frontrunning, an issue that is plaguing operational systems. 
By adapting~\alg, we allow it to be used for Proof of Work, preventing frontrunning by system operators and tuning the puzzle difficulty using the blockchain mechanism. 
\end{abstract}

\section{Introduction}

Various applications like auctions~\cite{rabin2006, parkes2008, bogetoft2006practical}, exchanges~\cite{bentov2019tesseract, Stellar, tex, idex11} and games~\cite{kumaresan2016amortizing,bentov2017instantaneous} require multiple users to each commit to a value.  
Critically, each user should learn nothing about the other commitments before making her own.  
%
%
The archetypal example~\cite{boneh2000timed, parkes2008, bogetoft2006practical} is an auction. 
In a second-price auction, where the top bidder wins and pays the second highest price, if the last bidder saw the other bids, she could propose a price negligibly smaller than the highest bid and maximize the seller's profit at the expense of the top bidder.
This is called \emph{frontrunning}. 

To prevent frontrunning, participants can use a trusted third party (TTP) to collect all commitments and publish them only after all participants have committed. 
But if the third party was compromised and collaborated with one of the participants, they could trivially perform frontrunning. 
Roughly, the question we address is how to design a protocol that doesn't use a TTP and allows a large number of participants to commit to values such that no one can perform frontrunning. 

Previous work has addressed various variants of this theme, but none solves the problem at hand. 
The TTP can be replaced by a distributed system~\cite{rabin2006, bellare1997verifiable}; but the system is vulnerable to frontrunning if enough elements of the distributed system are compromised.
Serial time-lock puzzles~\cite{Rivest1996,boneh2000timed} allow an encrypter to time-lock a single message such that decryption requires a long serial computation; but this approach does not scale to many participants. 
Other approaches require additional assumptions like trusted hardware~\cite{bentov2019tesseract} or synchrony at the commitment phase~\cite{tex}. 

In this paper we introduce the \emph{\problemlong}~(\emph{\problem}) problem. 
A set of~$N$ users complete an interactive~\emph{commit} protocol with a single \emph{coordinator} to form an arbitrary order of commitments on messages. 
Once done, the coordinator reveals the list of messages within a bounded expected time, independent of the number of users~$N$. 
Any third-party can verify the messages and that their order matches the commitments. 

A solution should be secure in face of a static \emph{PPT} adversary that can compromise the coordinator and all participants but one honest party and under the \emph{Discrete Logarithm~(DL)} and \emph{Random Oracle} assumptions and an assumption on the total computational power of the parties. 
Specifically, it should be \emph{sound}, that is, the adversary cannot reveal a message that was not committed. 
It should also be \emph{frontrunning-resistant}, that is, the adversary cannot affect the output based on an honest party's message. It should provide \emph{time-lock}, that is, the time it takes to reveal the commitments is a predefined parameter, meaning that an adversary cannot abort the protocol based on committed messages' content before that time. 
A solution should also provide an \emph{ensured output} guarantee, that is, a party that invests her resources in order to reveal the messages is guaranteed to succeed and output the original messages.
 
We present the \emph{Time-Capsule} protocol that solves the \problem problem. 
In~\alg the parties first jointly choose a random public ElGamal~\cite{Elgamal1985} moderately-hard encryption key~\cite{dwork1992pricing, Rivest1996, back2002hashcash}. 
Due to the ElGamal key-generation scheme, for each randomly chosen public key there exists a matching private key, but at this stage no entity in the system knows it.
The parties subsequently all commit to the order of their encrypted messages, locking down the order without relying on synchrony assumptions. 
Only once each party received all its peers' commitments, does it send its ciphertext.
To simplify the protocol, we employ an untrusted coordinator, which is responsible to coordinate protocol progress and aggregate data. 
Given all ciphertexts, the coordinator decrypts all messages by solving a DL puzzle whose solution is the private key.

The time to reveal the commitments is therefore determined by the parallel computational power available to the coordinator.
The estimated time to solve~DL should be calculated considering the best known algorithm.
A coordinator could in principle perform a precomputation (e.g., with Index calculus~\cite{joux2013new}) to reduce her decryption time. 
In order to improve the accuracy of time estimation, we propose a simple modification to the~DL problem definition.
We start with the observation that precomputation in known efficient solutions to the~DL problem~\cite{adrian2015imperfect, adleman1979subexponential} apply only to a single group and generator. 
We formalize this as an assumption that precomputation can only reduce effort for a polynomial set of groups. 
Based on this assumption we define the \emph{Precomputation-Resistant Discrete Logarithm} (\emph{\pcrdl}) problem and show it is as hard as a DL problem without precomputation. 
\pcrdl may thus be of independent interest for any application that depends on the hardness of the DL problem in~$\mathbb{Z}_p^*$ for its security, e.g., for ElGamal encryption and signature, Diffie Hellman key exchange and Schnorr signature, since using  \pcrdl instead of traditional DL reduces the probability to perform a precomputation attack on these schemes and reveal the private key of a party to be negligible.
 
Finally, we show how \alg can be adapted to achieve an advancement in the study of \emph{Proof of Work} (\emph{PoW})~\cite{dwork1992pricing,jakobsson1999proofs}  \emph{permissionless} protocols. 
Dwork and Naor~\cite{dwork1992pricing} and followup work presented PoW as a mechanism to increase the cost of spam~\cite{abadi2005moderately, dwork2003memory, back2002hashcash} and denial-of-service attacks~\cite{aura2000resistant,Waters2004}. 
The mechanism requires principals to prove they expended resources as a precondition to preforming actions like sending a message or an email. 
This imposes a minor cost on well-behaved users, but a high cost on an attacker who must solve numerous puzzles. 
Therefore, it establishes security in systems that do not require permission for entry. 
PoW is implemented as a \emph{cryptographic puzzle} that is quick to generate and verify but requires a tunable amount of resources to solve. 

Aspnes et al.~\cite{aspnes2005exposing} took the next step with a protocol that utilizes PoW to ensures participants in a classical consensus protocol cannot masquerade as multiple nodes in a so-called Sybil attack. 
Nakamoto~\cite{nakamoto2008bitcoin} made this approach practical by introducing the Bitcoin Blockchain PoW-oriented consensus protocol, including an organic mechanism for tuning the PoW difficulty, and an internal currency to incentivize participation. 
The blockchain operators, called \emph{miners}, create data structures called \emph{blocks} that contain user operations called \emph{transactions}. 
Each block must contain a valid proof of work, and rewards its creator with cryptocurrency tokens. 
The blockchain protocol was adopted in many so-called \emph{cryptocurrencies}~\cite{nakamoto2008bitcoin, buterin2013ethereum, litecoin, bonneau2015sok} to achieve \emph{censorship-resistance}. 
These cryptocurrencies facilitate payments among users, as well as so-called \emph{smart contracts} that enable elaborate applications like decentralized exchanges~\cite{etherdelta, idex11, radarrelay, uniswap}, auctions~\cite{galal2018verifiable, blass2018strain, galal2019trustee} and loans~\cite{aave, salt}. 
The underlying blockchain mechanism ensures censorship-resistance, i.e., no system operator can block user actions. 
Still, existing blockchains are vulnerable to frontrunning: each user's transactions are public, allowing others, particularly the system operators, to place their own transactions in advance or rearrange the order of transactions. 
This has become a practical concern in operational systems~\cite{daian2019flash,Submarine}. 

To overcome this we present \emph{\bname}, a frontrunning-resistant blockchain protocol, with a different user-miner communication model and a new PoW component.
\bname uses a variant of \alg to hide the contents of the transactions until they are placed in the blockchain. 
The decryption of the transactions serves as a PoW mechanism for the blockchain, and utilizes it as a practical way to tune the decryption hardness. 
The challenge is that in \alg the puzzle is deterministically determined by the users. 
For \problem this is not a problem~-- it is in the users' interest to choose a random puzzle each time. 
For PoW this is not sufficient, as users could collaborate to repeatedly solve the same puzzle and win the blockchain rewards. 
Therefore, \bname puzzles are based on the blockchain structure as well as the content of the block's transactions. 
We show that this change achieves the requirements for PoW, as well as the \problem definition apart from ensured output, which can be enforced through external means. 

In summary, our main contributions are as follows. 
\begin{inparaenum}
	\item Formalization of the \problem problem~(\S\ref{sec:model}), 
	\item presentation of the Time-Capsule algorithm~(\S\ref{sec:algorithm}),
	\item proof that \alg solves \problem~(\S\ref{sec:correctness}),
	\item formalization of precomputation limitation assumption (\pcrdl) and introduction of \pcrdl puzzles~(\S\ref{sec:pcrdl}), 
	\item presentation of \bname, a frontrunning-resistant blockchain~(\S\ref{sec:bc}), and  
	\item proof that \bname solves \problem, \pcrdl and provides blockchain PoW requirements~(\S\ref{sec:bc:correctness}). 
\end{inparaenum}


    \section{Related Work}
    \label{sec:related} 
\vspace{4pt}
\paragraph{Puzzles}\label{sec:related time lock puzzles}

Several works solve the problem of committing to a value that will be revealed at a certain time in the future~\cite{Rivest1996,boneh2000timed,bellare1997verifiable} but these solutions does not scale to a large set of commitments.

Other works scale to multiple commitments~\cite{rabin2006,parkes2008} but require trust in $t$ out of $n$ parties.

Homomorphic time-lock puzzles \cite{malavolta2019homomorphic} is an algorithm to compute a function over a set of time-locked puzzles. The algorithm constructs a new time-locked puzzle from all of them; the solution to the constructed puzzle is the output of the function.
Both \alg and homomorphic time-lock puzzles offer a solution for hiding a set of secrets for a period of time, while a certain functionality is finalized.
However, Homomorphic time-lock puzzles' values and solutions are restricted to a finite size, thus they do not scale to reveal a large list of commitments. Moreover, Homomorphic time-lock puzzles assumes synchrony for the commitment phase.

\paragraph{Verifiable Delay Function (VDF)} \label{sec:related vdf}
VDF~\cite{Boneh2018, pietrzak2019simple, wesolowski2019efficient, boneh2018survey} is a function that requires a large number of sequential computations to evaluate and whose output can be verified efficiently. This asymmetry is similar to the puzzle we use in \alg.
However, in general, VDFs provide only time delay and do not hide any secrets during this delay and thus are not useful for \problem.
As opposed to VDF, Time-Capsule puzzle is not sequential; it is an interesting question whether it has a serial alternative.

\paragraph{MPC and ZKP}\label{sec:related mpc}
Multi-party computation (MPC) protocols~\cite{yao1986generate, goldreich1987how,chaum1988multiparty} enable a set of un-trusting parties to compute a function of their private inputs, learning only the result of a function and nothing else about the inputs.

Unlike MPC, \problem requires that the inputs remain private only for a limited time.
Typical practical MPC solutions, for example~\cite{bogetoft2009secure, bogetoft2006practical,damgaard2012multiparty, ben2008fairplaymp}, are designed for a specific application or requirement, such as semi-honest parties, $t$ out of $n$ trusted parties or a small number of participant, whereas \alg is a general protocol for any application and any number of participants in a trustless environment.

Some MPC applications, such as auctions and coin tossing, that require the inputs to stay private only until the set of inputs is finalized, can be solved by \alg.

Zero Knowledge Proofs (ZKP)~\cite{blum2019non,ben2014succinct, blum1991noninteractive, goldreich1994definitions} allow to prove what is the output of a specific function, without revealing anything else about the input. ZKP does not solve \problem, which requires to reveal all commitments after a time delay.

\paragraph{Permissionless Blockchains}\label{sec:related pow}
We proceed to discuss blockchain work related to \bname.

Many works \cite{zhang2017rem, tsabary2019heb, gilad2017algorand, kiayias2017ouroboros, bentov2014proof, ball2017proofs, karantias2019proof, eyal2016bitcoin} address the wastefullness of blockchain PoW.
The \bname PoW puzzle provides a timely decryption service, thus being less wasteful than common PoW computation.

Privacy coins~\cite{sasson2014zerocash, miers2013zerocoin, zcash, duffield2018dash, monero} are blockchains that hide the contents of transactions. These blockchains do not solve \problem because their privacy does not allow to run on-chain smart contracts, which requires the participating parties to compute the smart contract's result for each input and know its value. 

Several works implement frontrunning-resistant auctions, exchanges and general applications over smart contracts~\cite{Asayag, bentov2019tesseract, tex, kosba2016hawk, galal2018verifiable, blass2018strain, galal2019trustee,Submarine, breidenbach2018enter}. They rely on stronger assumptions than in \problem, such as penalties, trusted execution environment, trust in an honest majority or in $t$ out of $n$ parties.

\paragraph{Blockchain Applications}\label{sec:related exchanges}
Analyses of frontrunning in cryptocurrency blockchains \cite{Eskandari, daian2019flash} report a sizable economy of bots profiting from frontrunning and arbitrage in auctions and Decentralized exchanges (DEXes).
Decentralized Finance (DeFi)~\cite{defi} is a variety of decentralized financial services and DEXes~\cite{idex11,uniswap, bancor, radarrelay, etherdelta}, typically implemented with Ethereum smart contracts. 
These services operate in a non-custodial manner, without a trusted third party holding the users' funds; they are thus protected from a malicious service provider who tries to steal the users' funds. However, these services are susceptible to frontrunning and would benefit from a frontrunning-resistant blockchain.

\paragraph{Second Layer Blockchain Applications}\label{sec:related applications}
Several works use a blockchain to enforce fairness of games and contracts \cite{kumaresan2015use,kumaresan2016amortizing,andrychowicz2014fair,andrychowicz2014secure,bentov2017instantaneous}. 
They use a blockchain as a trusted party to collect collateral from the users and use it as penalty in case they misbehave. 
However, blockchains are operated by independent nodes which can be bribed~\cite{mccorry2018smart, judmayer2019pay} to violate the expected behavior, making them unreliable.
\bname does not assume trusted third party.

Second-tier protocols~\cite{poon2016bitcoin, poon2017plasma, dziembowski2019perun, dziembowski2018general, lind2019teechain} allow a set of parties to place collateral on a blockchain, allowing any of them to terminate the protocol by placing a transaction on the blockchain that redistributes the amount in the collaterals. If the terminator acts maliciously, another party can dispute its transaction within a long time period.
The schemes are secure over a censorship-resistant blockchain and do not require frontrunning-resistance.

	\section{\problem Model} 
	\label{sec:model}
		\subsection{Structure and Assumptions}                    

The system comprises a set of $N$ participants $P=\{P_1, P_2,....,P_N\}$ and a coordinator. Each participants~$P_i$ has a secret key $Sk_i$ and access to a public key infrastructure (PKI) service that maintains the matching public keys $\{Pk_1,...,Pk_N\}$. All parties can exchange messages via reliable channels. The system is parametrized by two security parameters $\lambdam$ and $\lambdah$, for \emph{computationally moderate} and \emph{computationally hard} problems, respectively.
Each participant has a private memory of size polynomial in $\lambdam$ and a known parallel computational power. 
A static adversary, PPT in~$\lambdah$, can compromise the coordinator and all parties except for one. 

\paragraph{Cryptographic assumptions}
We assume that the DL assumption over the multiplicative group modulo~$p$ holds: Let~$p$ be a safe prime of size~$\lambdam$ bits. Let~$\mathbb{G}$ be the multiplicative group modulo $p$ generated by~$g \in \mathbb{G}$. Let $a$ be chosen uniformy at random~$a  \randselect Z^{\ast}_P$. Given~$(\mathbb{G}, p, g, g^a)$, a PPT adversary's probability of finding $a$ is negligible in~$\lambdam$.
Given a time~$\tau$, we can calculate the matching~$\lambdam$ value such that the DL puzzle with~$\lambdam$ as security parameter will take an expected mean time~$\tau$ to solve.
The participants have a \emph{difficulty table} that maps each~$\lambdam$ value to a prime number~$p$ and a generator $g$.
  
We also define and assume the Multi-DL assumption as follows:
\begin{assumption}[Multi-DL assumption]
	Let~$k$ be a number polynomial in~$\lambdam$. For any PPT adversary~\Ad, solving one out of~$k$ random DL puzzles, the advantage over solving a single DL puzzle is negligible.
\label{def:multi dl assumption}
\end{assumption} 

We assume a secure signature scheme with operations \textit{sign} and \textit{verifySig}. This scheme guarantees a signature forgery probability~$\varepsilon_{\textit{sig-forgery}}$, which is negligible in~$\lambdah$.
Each party has access to a hash function, modeled as a random oracle~$\fun{hash}:\{0,1\}^*\rightarrow\{0,1\}^{\lambdah}$. The probability to find the pre-image of the hash is~$\varepsilon_{\textit{pre-image}}$ and the probability of a hash collision is~$\varepsilon_{\textit{hash\_collision}}$
Participants can use a Merkle tree~\cite{merkle1980protocols} structure, which guarantees a negligible root collision probability~$\varepsilon_{\textit{MT\_collision}}$.

Additionally, there is a pseudo-random number generator (PRNG) that provides a function~$\fun{PRNG}(\textit{seed}, \textit{first\_bit}, \textit{bit\_count})$ and returns \textit{bit\_count} bits, starting from \textit{first\_bit} of the pseudo random bits generated from the randomness \textit{seed}. The PRNG guarantees that a polynomial length output is computationally indistinguishable from a random output with an advantage~$\varepsilon_{\textit{PRNG}}$, negligible in~$\lambdah$. This assumption can be realized, e.g., with the Blum-Micali $n$-wise sequential composition~\cite{boneh2017graduate}. Let $G$ be a secure PRG that holds the advantage~$\varepsilon_\textit{PRG}$, negligible in~$\lambdah$.
The Blum-Micali~$n$-wise sequential composition~$G'$ expands~$G$'s output from~$t + l$ bits to~$n\times t + l$ bits and is also a secure PRG with the following advantage:
\begin{equation}
\label{eq:blum_micali}
\textit{PRGadv}[\mathcal{A},G']=n*\textit{PRGadv}[\mathcal{A},G] = n*\varepsilon_{PRG},\nonumber
\end{equation}	
which is negligible for a polynomial number of iterations~$n$. A practical choice for PRNG is ChaCha~\cite{bernstein2008chacha}.

\setlength{\floatsep}{7 pt}
\setlength{\textfloatsep}{7 pt}
\begin{figure}[t]
	\centering
	\begin{tikzpicture}[node distance=0cm,auto,>=latex']
	\begin{scope}
	\node[xshift=-1cm](stub){};
	\node[fit={(0,0) (0.4,0.25)}, inner sep=0pt,xshift=0cm, yshift=1.00cm] (PN) {
		$P_N$
	};
	\node[fit={(0,0) (0.4,0.25)}, inner sep=0pt,xshift=0cm, yshift=-1.25cm] (PNe) {
	};
	\node[fit={(0,0) (0.4,0.25)}, inner sep=0pt,xshift=0.65cm, yshift=1.00cm] (Pd) {
		$\dots$
	};
	\node[fit={(0,0) (0.3,0.25)}, inner sep=0pt,xshift=1.5cm, yshift=1.00cm] (P1) {
		$P_1$
	};
	\node[fit={(0,0) (0.3,0.25)}, inner sep=0pt,xshift=1.5cm, yshift=-1.25cm] (P1e) {
	};
	\node[fit={(0,0) (1.6,0.25)}, inner sep=0pt, xshift=3.4cm, yshift=1.00cm] (Pc) {
		Coordinator
	};
	\node[fit={(0,0) (1.6,0.25)}, inner sep=0pt,xshift=3.4cm, yshift=-1.25cm] (Pce) {};
	\node[fit={(0,0) (0,0)}, inner sep=0pt, xshift=4.15cm, yshift=0.75cm] (Pcom1c) {};
	\node[fit={(0,0) (0,0)}, inner sep=0pt, xshift=1.7cm, yshift=0.75cm] (Pcom1) {};
	\node[fit={(0,0) (0,0)}, inner sep=0pt, xshift=4.15cm, yshift=0.15cm] (PcomNc) {};
	\node[fit={(0,0) (0,0)}, inner sep=0pt, xshift=0.25cm, yshift=0.15cm] (PcomN) {};
	\node[fit={(0,0) (0.25,0.25)}, inner sep=0pt, xshift=3cm, yshift=0.15cm] (Pcomd) {$\vdots$};
	\node[fit={(0,0) (0,0)}, inner sep=0pt, xshift=4.1cm, yshift=-0.1cm] (Pcomend) {};
	\node[fit={(0,0) (0,0)}, inner sep=0pt, xshift=4.1cm, yshift=-.85cm] (Prevc) {};
	\node[fit={(0,0) (1,0.2)}, inner sep=0pt, xshift=3.15cm, yshift=-1cm] (out) {\footnotesize output};
	
	\path[-] (PN) edge node {} (PNe);		
	\path[-] (P1) edge node {} (P1e);	
	\path[-] (Pc) edge node {} (Pce);
	\path[<->] (Pcom1c) edge node [above]{} (Pcom1);
	\path[<->] (PcomNc) edge node [above]{} (PcomN);
	\path[<->] (Pcomend) edge node [left]{\footnotesize $\tau$} (Prevc);
	
	\draw [decorate,decoration={brace,amplitude=4pt},yshift=0pt]
	(4.25,0.85) -- (4.25,0) node [black,midway,xshift=0.2cm] {\footnotesize
		commit};
	
	\draw [decorate,decoration={brace,amplitude=4pt},yshift=0pt]
	(4.25,-0.05) -- (4.25,-0.9) node [black,midway,xshift=0.2cm] {\footnotesize
		reveal};
	\end{scope}	
	\end{tikzpicture}
	\caption{\problem Diagram}
	\label{pmptc diagram}
\end{figure}

		\subsection{Goal}                                                                  
A protocol~$S$ solving the \problem comprises three elements: a \emph{commit} protocol, a \emph{reveal} function and a \emph{verify} function.
It proceeds in two phases: commit and reveal (Figure \ref{pmptc diagram}).

In the commit phase, each participant~$P_i$ in $P$ chooses a message  $m$ and runs an interactive protocol with the coordinator to commit this message and agree on a list of message commitments. In the reveal phase, the coordinator runs reveal($\cdot$), which proceeds for a random time with mean value $\tau$ and then outputs a list $L_m$ of tuples $(m, i)$  and a proof \textit{proof}. Anyone can verify that the output is correct using the function verify($L_m$, \textit{proof}), which returns the value TRUE if and only if $L_m$ matches the messages and order confirmed in the commit phase. If verify($L_m$, \textit{proof}) = TRUE, we say that ($L_m$, \textit{proof}) is valid.
We denote the tuple $(m,i)$ by $m^i$.

We formalize the \problem security as the following properties:
	\paragraph{Soundness} The coordinator outputs $m^i$ only if the participant $P_i$ committed to $m$ in the commit phase.
	We formally define this property using the \gamee{Sound}{S} game played between a challenger that represents the participants and an adversary that represents the coordinator. 
	We denote a multiset by \{\{$\cdot$\}\} and multiset subtraction by~$\backslash$.
	
	\mytextbox{
	\begin{itemize}
		\item The adversary chooses a multiset of messages $\{\{m_1,....,m_k\}\}$, $k$ is polynomial in $\lambdah$;
		\item the challenger and adversary run the commit protocol: The challenger has~$k$ different identities~$P_\textit{challenger}=\{P_1,..,P_k\}$; for each message~$m_j \in \{\{m_1,....,m_k\}\}$, $P_j$ commits~$m_j$.
		\item The adversary outputs $(L_m, \textit{proof})$.
	\end{itemize}

	The adversary wins the \gamee{Sound}{S} game if she outputs a valid output such that $\{\{m_i^j \in L_m|P_j \in P_\textit{challenger}\}\} \backslash \{\{m_1^1,....,m_k^k\}\} \neq \emptyset$. In other words, the output contains a message as if it was committed by the challenger, even though it was not.
	We denote the event that the adversary wins by $W^{S}_{\textit{Sound}}$. We define the advantage in this game by  $\text{Sound}\textsf{adv}[\lambdah, \Ad]\triangleq Pr[W^{S}_{\textit{Sound}}]$.
}
	\begin{definition}[Soundness]
		A protocol $S$ provides \emph{Soundness} if for any PPT adversary \Ad playing the \gamee{Sound}{S} game, $\text{Sound}\textsf{adv}[\lambdah, \Ad]$ is negligible in $\lambdah$.
	\end{definition}

	\paragraph{Frontrunning-resistance}	FR-resistance means that an adversary cannot change the content of its output based on the content of messages generated by other parties. Specifically, she cannot add (or equivalently, remove) a specific message to the list or change the order of the list. We formalize these two requirements using two games, which we describe below using the following template:

\mytextbox{
	\begin{itemize}
		\item The adversary chooses three different messages $m_A, m_0, m_1$ and sends $m_0$ and $m_1$ to the challenger;
		\item the challenger chooses $b \randselect \{0,1\}$;
		\item the challenger and adversary run the commit protocol: The challenger, as particiant $P_i$, commits a message $m_b$ and the adversary commits $k$ messages of her choice, where $k$ is polynomial in $\lambdah$.
		\item The adversary outputs $(L_m, \textit{proof})$.
	\end{itemize}
}

	In the first game, the \emph{Frontrunning Message Choosing} (\emph{FMC}) game, the adversary wants a message to appear depending on the choice of the challenger. The adversary wins if the output is valid and (1)~if $b=0$, then one of the messages in $L_m$ is $m_A$ or (2)~if $b=1$, then none of the messages in $L_m$ is $m_A$.
	We denote the event that the adversary wins by $W^{S}_{\textit{FMC}}$.
	Obviously, the adversary can win the game with probability $1/2$ by choosing whether to commit~$m_A$ uniformly at random. We define the advantage of an adversary \Ad in this game by $\text{FMC}\textsf{adv}[\lambdam, \lambdah, \Ad]\triangleq|Pr[W^{S}_{\textit{FMC}}]-1/2|$.
	
	In the second game, the \emph{Frontrunning Order choosing} (\emph{FOC}) game, the adversary's goal is to order the messages according to the choice of the challenger. The adversary wins the game if the output is valid, only one of the messages~$m_0$ or~$m_1$ appears in~$L_m$ and (1)~if~$b=0$, then~$m_A$ appears before~$m_0^i$ in~$L_m$ or (2)~if~$b=1$, then~$m_A$ appears after~$m_1^i$.
	We denote the event that the adversary wins by~$W^{\textit{S}}_{\textit{FOC}}$.
	Again, the adversary can win the game with probability~$1/2$ if the order of the list is set uniformly at random. We define the advantage of adversary~\Ad in this game by~$\text{FOC}\textsf{adv}[\lambdam, \lambdah, \Ad]\triangleq|Pr[W^{S}_{\textit{FOC}}]-1/2|$.

	With these two games we can define FR-resistance.
	\begin{definition}[FR-Resistance]
	A protocol $S$ provides \emph{FR-resistance} if for any PPT adversary \Ad playing either \gamee{FMC}{S} or \gamee{FOC}{S}, the advantages  $\text{FMC}\textsf{adv}[\lambdam,\lambdah,\Ad]$ and $\text{FOC}\textsf{adv}[\lambdam,\lambdah,\Ad]$ respectively, are negligible in $\lambdah$.
	\label{def:pmptc fr resistance}
	\end{definition}

	\paragraph{Time-Lock} If at least one participant follows the protocol, the time it takes to reveal a single commitment committed by a participant that follows the protocol has a mean value of~$\tau$. The reveal delay time can start at the beginning of the commit phase, thus the mean reveal time is~$\tau$.
	This requirement applies for each commitment submitted by a participant that is following the protocol.
	
	\begin{definition}[Time-Lock]
		Let $\tau$ be a parameter defined in time units. For a run of the protocol, we mark the beginning of the commit phase by time~$0$ and the time that the commitment of participant $P_i$ is revealed by~$t_\textit{reveal}^i$. We denote the set of participants following the protocol by $P_H \in P$. If $|P_H| \geq 1$ then the mean value of~$t_\textit{reveal}^j$ where $P_j\in P_H$ is at least~$\tau$.
	\label{def:pmptc time lock}
	\end{definition}
	Note that the mean reveal time~$\tau$ is independent of the number of participants~$N$.

	\paragraph{Ensured Output} If the coordinator is following the protocol then she's guaranteed to be able to produce a valid output after the declared time delay.
	\begin{definition}[Ensured Output]
		Once the commit phase has ended, a coordinator that is following the protocol will output a valid output after mean time $\tau$.
	\end{definition}

\paragraph{\problem}
Using these properties we define the \problem problem.
\begin{definition}[\problem]
	A tuple (\emph{commit}, \emph{reveal}, \emph{verify}) of a commit protocol~$S$ \emph{commit}, a reveal function \emph{reveal} and a verify function \emph{verify} that provides Soundness, FR-Resistance, Time-Lock and Ensured Output \emph{solves \problem}.
\label{def:pmptc_protocol}
\end{definition}

    \section{\alg Algorithm} 
    \label{sec:algorithm}
 
We present \emph{\alg}, an algorithm that solves the \problem problem. 
The crux of the approach is as follows. 
The parties jointly form a moderate-difficulty public encryption key, without anyone in the system knowing the matching secret key. 
Each party encrypts its message and sends it to the coordinator. 
The coordinator uses a brute-force search to decrypt the messages, which succeeds after a predictable time and allows her to publish the messages. 
To prevent attacks where the coordinator manages by luck to find the secret key before the commitment phase is done, the parties actually start by committing to their encrypted messages. 
We proceed to describe the commit protocol and reveal and verify algorithms.

		\paragraph{Commit}
The commit protocol between each participant $P_i$ and the coordinator proceeds as follows (Figure \ref{commit protocol}).
Each participant chooses a random nonce (line \refp{line:gen nonce}) and sends it to the coordinator. The coordinator aggregates the nonces in a Merkle tree $T_\textit{seed}$ (\refc{line:gen merkle tree}). The root of the tree will serve as a random seed (\refc{line:seed is ready}). The coordinator sends \textit{seed} to each participant as well as a Merkle proof (\refc{line:nonce merkle proof}) that her nonce is included in the tree for her to verify~(\refp{line:verify_mp}). In our pseudo code, assert means that if the check fails, the users aborts.

Next, each participant calculates an ElGamal encryption public key for the Time-Lock $Pk_\textit{TL}$ from the $\textit{seed}$ (\refp{line:gen puzzle}).
The public key for ElGamal encryption has three parameters: a prime number $p$, a generator $g$ of the group $Z^*_p$ and a member $b$ of the group $Z^*_p$. 

We need to be able to tune the puzzle's difficulty according to the computation power in the network to reach the delay target. We use the group size $|Z^{\ast}_P|= |p-1|$ to adjust the difficulty of the puzzle. This means that for each difficulty we choose a different prime number for the puzzle.

The public key calculation is done using the deterministic function $\fun{GenPuz}(\lambda, \textit{seed})$ (Algorithm~\ref{alg:generate_pk}). The function takes a security parameter~$\lambdam$ and randomness~$\textit{seed}$. First, it selects a safe prime~$p$ and a generator~$g$ from the difficulty table (line \ref{line:choose p g}).

Then it chooses the group member~$b$ with iteration sampling as follows. 
It expands the seed to a longer pseudo random bit stream by using the PRNG.
It constructs a candidate $b_\textit{candidate}$ from the pseudo random bits and checks whether it is in the group $Z^{\ast}_P$; this is repeated until a valid candidate is found~(lines \ref{line:genpuz_b_search_start}-\ref{line:genpuz_b_search_end}). The function returns the three values $Pk_\textit{TL} \gets (p,g,b)$.

Getting back to the commit protocol (Figure~\ref{commit protocol}), now that the participant has calculated~$Pk_\textit{TL}$, she chooses a random value $r_i$ (\refp{line:gen_r}), concatenates the message $m_i$ and $r_i$ and encrypts it using ElGamal encrytion (\refp{line:enc m}). 
She computes a hash of the encrypted value (\refp{line:query random oracle}). She sends the hash output and her index $i$ to the coordinator.
The coordinator creates a list~$L_H$~(\refc{line:gather Hc}, \refc{line:gen_list}) of all the values she received from the participants, these are the commitments of the encrypted messages. She sends the list to all of the participants.

Each participant verifies that her commitment appears once in $L_H$ (\refp{line:verify appears once}), otherwise she aborts. She then signs the list and \textit{seed} with her secret key $Sk_i$ (\refp{line:sign_list}) and sends the signature along with the encrypted message, $C_{m_i}$, and her index $i$ to the coordinator. 
The coordinator verifies that the signature is valid~(\refc{line:verify p sig}) and that $C_{m_i}$ matches the previously received commitment~(\refc{line:verify p cipher}) in order to guarantee the validity of the output, otherwise she aborts. 

Once the coordinator collected all encrypted messages and signatures~(\refc{line:collect sigs}, \refc{line:collect C}), the commit phase is done.

		\paragraph{Reveal}
Once the coordinator has received the signatures from all participants, she can proceed to reveal the commitments by finding the decryption key $Sk_\textit{TL}$ corresponding to $Pk_\textit{TL}$. She does that by searching for $a \in Z^{\ast}_P$ such that $b=g^a \text{ mod } p$, namely solving an instance of the DL problem. The coordinator uses an efficient algorithm of her choice to solve the discrete logarithm problem \cite{pollard1978monte, pohlig1978improved, joux2013new, corrigan2018discrete} and the difficulty table is tuned for such an efficient algorithm. But to make things concrete we consider for example a na\text{\"i}ve linear search implementation (Algorithm \ref{alg:find_elgamaml_sk}). 

Once the coordinator finds $Sk_\textit{TL}$, the puzzle solution, she decrypts the commitments and outputs them along with the key $Sk_\textit{TL}$, the signature list $L_s$ and the \textit{seed} (Algorithm~\ref{alg:coordinator reveal}).

		\paragraph{Verify}
Any participant can validate the coordinator's output with the function \fun{VerifyOutput}($\cdot$) (Algorithm \ref{alg:verify_reveal}). First, the participant verifies that $Sk_\textit{TL}$ matches $Pk_\textit{TL}$. Next, she uses the encryption key $Pk_\textit{TL}$ to encrypt each message and generates the list of hashes of the encrypted messages. And finally, she verifies that each signature matches this list and \textit{seed}, using the PKI.

		\paragraph{Implementation variations}
Note that in \alg, the order of the committed messages' list is set in the commit protocol (\refc{line:gen_list}). However, there are other options. The order can be determined at an earlier stage, according to the order of nonces in the merkle tree (\refc{line:gen merkle tree}); it can also be determined later, after the revel phase, by permuting the list according to a random seed constructed from the revealed messages.

Another interesting question is whether a puzzle alternative can work in EC groups where different difficulty choices imply different curve choices. We leave that for future work.

\begin{figure*}
	\footnotesize
	\begin{tikzpicture}[auto,>=latex']
	\begin{scope}
	\node[fit={(0,0) (0.25,0.25)}, inner sep=0pt,xshift=6.25cm, yshift=1.25cm] (Pi) {
		{$\mathbf{P_i}$}
	};
	\node[fit={(0,0) (0.25,0.25)}, inner sep=0pt,xshift=6.25cm, yshift=-4.2cm] (Pie) {
	};
	
	\node[fit={(0,0) (1.6,0.25)}, inner sep=0pt, xshift=8.85cm, yshift=1.25cm] (Pd) {
		\textbf{Coordinator}
	};
	
	\node[fit={(0,0) (1.6,0.25)}, inner sep=0pt,xshift=8.85cm, yshift=-4.2cm] (Pde) {};
	
	\algrenewcommand{\alglinenumber}[1]{\footnotesize p{#1}:}
	\renewcommand{\Statex}{\item[\hphantom{\footnotesize p\arabic{ALG@line}:}]}
	
	\node[fit={(0,0) (6,1)}, inner sep=0pt,yshift=0.55cm] (I1) {
		\begin{algorithmic}[1]
		\State $\textit{nonce}_i \randselect \{0,1\}^{\lambdah}$
		\label{line:gen nonce} 
		\algstore{I}
		\end{algorithmic}
	};

	\node[fit={(0,0) (6,2.5)}, inner sep=0pt,yshift=-1.7cm] (I2b) {
		\begin{algorithmic}[1]
		\algrestore{I}
		\Statex \textcolor{gray}{Encrypt and Commit}
		\State {\label{line:verify_mp} assert(verifyMP(${seed}, \textit{MP}_i, \textit{nonce}_i$)})
		\State $Pk_\textit{TL} \gets $ \fun{GenPuz}($\lambdam$, seed)
		\label{line:gen puzzle}
		\State $r_i \randselect \{0,1\}^{\lambdah}$
		\label{line:gen_r}
		\State $C_{m_i} \gets \fun{enc}_{Pk_\textit{TL}}(m_i \| r_i)$
		\label{line:enc m}
		\State $H_{C_{m_i}} \gets \fun{hash}(C_{m_i})$
		\label{line:query random oracle}
		\algstore{I}
		\end{algorithmic}
	};
	\node[fit={(0,0) (6,1.5)}, inner sep=0pt, yshift=-2.9cm] (I3) {
	\begin{algorithmic}[1]
	\algrestore{I}
	\State assert({$H_{C_{m_i}} \text{ appears once in } L_H$})
	\label{line:verify appears once}
	\State $S_{{L_H},i}=\fun{sign}_{Sk_i}(L_H \parallel \textit{seed})$
	\label{line:sign_list}
	\end{algorithmic}
};
	\algrenewcommand{\alglinenumber}[1]{\footnotesize c{#1}:}
	\renewcommand{\Statex}{\item[\hphantom{\footnotesize c\arabic{ALG@line}:}]}

	\node[fit={(0,0) (7,2)}, inner sep=0pt,yshift=-0.5cm, xshift=10.75cm] (D1) {
		\begin{algorithmic}[1]
		\Statex \textcolor{gray}{Generate Seed}
		\Statex [Wait for $N$ nonces]
		\State Aggregate all nonces in $T_\textit{seed}$
		\label{line:gen merkle tree} 
		\State $\textit{MP}_i \gets$ Merkle proof that $\textit{nonce}_i$ is in $T_\textit{seed}$ 
		\label{line:nonce merkle proof}
		\State $\textit{seed} \gets T_\textit{seed}$'s root
		\label{line:seed is ready}
		\algstore{D}
		\end{algorithmic}
	};

	\node[fit={(0,0) (7,1.5)}, inner sep=0pt, yshift=-1.85cm, xshift=10.75cm] (D2) {
		\begin{algorithmic}[1]
		\algrestore{D}
		\Statex \textcolor{gray}{Form a Commitment List}
		\State $S_H \gets S_H \cup (H_{C_{m_i}}, i)$
		\label{line:gather Hc}
		\Statex [Wait until done for all parties]
		\State $L_H \gets \fun{PrepareList}(S_H)$  
		\label{line:gen_list}
		\algstore{D}
		\end{algorithmic}
	};

	\node[fit={(0,0) (7,2.85)}, inner sep=0pt, yshift=-4.35cm, xshift=10.75cm] (D3) {
		\begin{algorithmic}[1]
		\algrestore{D}
		\Statex \textcolor{gray}{Collect Ciphertexts and Signatures}
		\State assert({\fun{verifySig}($Pk_i$, $L_H \parallel \textit{seed}$, $S_{{L_H},i}$)})
		\label{line:verify p sig}
		\State $j \gets $ get index of participant $i$ in $L_H$
		\State assert({$L_H[j] = (\fun{hash}(C_{m_i}),i)$})
		\label{line:verify p cipher}
			\State $L_s[j] \gets S_{{L_H},i}$
			\label{line:collect sigs}
			\State $L_C[j] \gets (C_{m_i}, i)$
			\label{line:collect C}
		\end{algorithmic}
	};
	
	\node[fit={(0,0) (0.25,0.25)}, inner sep=0pt,xshift=6.25cm, yshift=0.75cm] (I1p1) {};
	
	\node[fit={(0,0) (0,0)}, inner sep=0pt,xshift=0cm, yshift=0.7cm] (I1Ls) {};
	\node[fit={(0,0) (0,0)}, inner sep=0pt,xshift=6.25cm, yshift=0.7cm] (I1Le) {};
		
	\node[fit={(0,0) (0.25,0.25)}, inner sep=0pt,xshift=9.5cm, yshift=.5cm] (D1p1) {};
	\node[fit={(0,0) (0.25,0.25)}, inner sep=0pt,xshift=9.5cm, yshift=0.25cm] (D1p2) {};
	
	\node[fit={(0,0) (0,0)}, inner sep=0pt,xshift=9.75cm, yshift=-0.45cm] (D1Ls) {};
	\node[fit={(0,0) (0,0)}, inner sep=0pt,xshift=17cm, yshift=-0.45cm] (D1Le) {};
	
	\node[fit={(0,0) (0.25,0.25)}, inner sep=0pt,xshift=6.25cm, yshift=-0.5cm] (I2p1) {};
	\node[fit={(0,0) (0.25,0.25)}, inner sep=0pt,xshift=6.25cm, yshift=-0.75cm] (I2p2) {};
	
	\node[fit={(0,0) (0,0)}, inner sep=0pt,xshift=0cm, yshift=-1.7cm] (I2Ls) {};
	\node[fit={(0,0) (0,0)}, inner sep=0pt,xshift=6.25cm, yshift=-1.7cm] (I2Le) {};
	
	\node[fit={(0,0) (0.25,0.25)}, inner sep=0pt,xshift=9.5cm, yshift=-1.25cm] (D2p1) {};
	\node[fit={(0,0) (0.25,0.25)}, inner sep=0pt,xshift=9.5cm, yshift=-1.5cm] (D2p2) {};
	
	\node[fit={(0,0) (0,0)}, inner sep=0pt,xshift=9.75cm, yshift=-1.95cm] (D2Ls) {};
	\node[fit={(0,0) (0,0)}, inner sep=0pt,xshift=17cm, yshift=-1.95cm] (D2Le) {};
	
	\node[fit={(0,0) (0.25,0.25)}, inner sep=0pt,xshift=6.25cm, yshift=-2.25cm] (I3p1) {};
	\node[fit={(0,0) (0.25,0.25)}, inner sep=0pt,xshift=6.25cm, yshift=-2.4cm] (I3p2) {};
	
	\node[fit={(0,0) (0.25,0.25)}, inner sep=0pt,xshift=9.5cm, yshift=-3.15cm] (D3p1) {};
	
	\path[-] (Pi) edge node {} (Pie);	
	\path[-] (Pd) edge node {} (Pde);
	\path[->] (I1p1) edge node[above, midway, sloped] {$\textit{nonce}_i$} (D1p1);

	\path[->] (D1p2) edge node[above, midway, sloped] {(seed,$\textit{MP}_i$)} (I2p1);
	\path[->] (I2p2) edge node[above, midway, sloped] {($H_{C_{m_i}}, i$)} (D2p1);
	\path[->] (D2p2) edge node[above, midway, sloped] {$L_H$} (I3p1);
	\path[->] (I3p2) edge node[above, midway, sloped] {$(S_{{L_H},i}, C_{m_i}, i)$} (D3p1);

	\path[-] (I1Ls) edge[dashed] node {} (I1Le);	
	\path[-] (I2Ls) edge[dashed] node {} (I2Le);
	\path[-] (D1Ls) edge[dashed] node {} (D1Le);	
	\path[-] (D2Ls) edge[dashed] node {} (D2Le);

	\draw [decorate,decoration={brace,amplitude=8pt},yshift=0pt] 
	(6,0.6) -- (6,-1.6) node {};
	
	\draw [decorate,decoration={brace,amplitude=8pt},yshift=0pt] 
	(6,0.-1.8) -- (6,-2.8) node {};

	\draw [decorate,decoration={brace,amplitude=8pt},yshift=0pt]
	(10.03,-0.35) -- (10.03,1.15) node {};
	
	\draw [decorate,decoration={brace,amplitude=8pt},yshift=0pt]
	(10.03,-1.85) -- (10.03,-0.55) node {};
	
	\draw [decorate,decoration={brace,amplitude=8pt},yshift=0pt]
	(10.03,-4) -- (10.03,-2.05) node {};
	\end{scope}
	\end{tikzpicture}	
	\caption{Commit Protocol}
	\label{commit protocol}
\end{figure*}
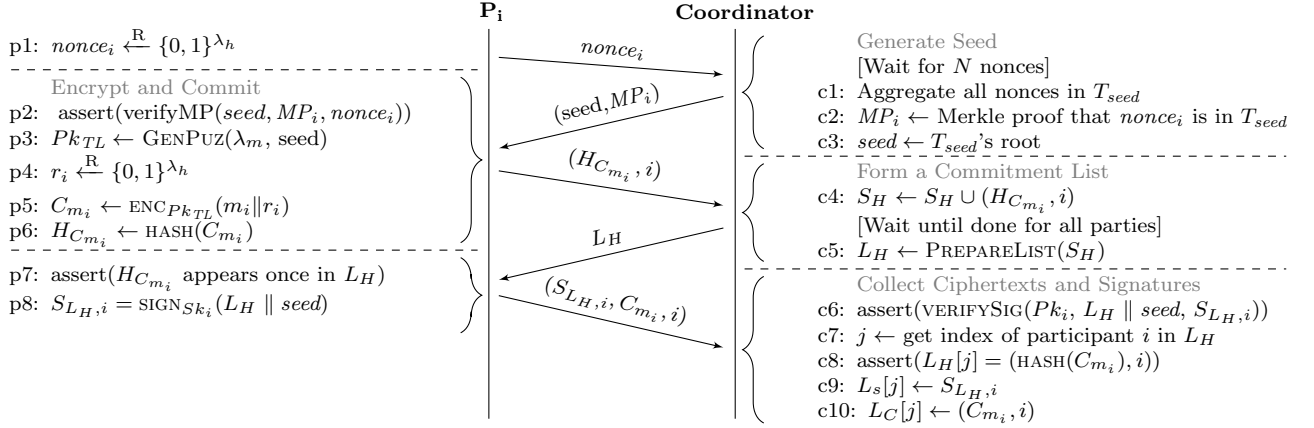

\begin{algorithm}[t]
	\begin{algorithmic}[1]
		\Function{GenPuz}{$\lambdam, seed$}
		\State $(p, g) \gets \fun{prime\_and\_generator\_from\_table}(\lambdam)$
		\label{line:choose p g}
		\State $j_\textit{bitnum} \gets 0$
		\label{line:genpuz_b_search_start}
		\State $b_\textit{candidate} \gets 0$
		\While {$b_\textit{candidate} \notin [2,p-1]$}
		\State $b_\textit{candidate} \gets \fun{PRNG}(\textit{seed}, j_\textit{bitnum}, \lambdam)$
		\State $j_\textit{bitnum} \gets j_\textit{bitnum} + \lambdam$
		\EndWhile
		\label{line:genpuz_b_search_end}
		\State \textbf{return} $(p, g, b_\textit{candidate})$
		\EndFunction
	\end{algorithmic}		
	\caption{Generate Puzzle}
	\label{alg:generate_pk}
\end{algorithm}

\begin{algorithm}[t]
	\begin{algorithmic}[1]
		\Function{Reveal}{\textit{seed}, $L_s$, $L_c$}
		\State $Pk_\textit{TL} \gets$ \fun{GenPuz}($\lambdam$, \textit{seed})
		\State $Sk_\textit{TL} \gets  \fun{find\_ElGamal\_Sk}(Pk_\textit{TL})$
		\label{line:find sk}
		\ForAll{$(C_m, i) \in  L_c$}
		\label{line:reveal cm to m start}
		\State $m \gets \fun{dec}_{Sk_\textit{TL}}(C_m)$
		\State $L_m \gets L_m \parallel (m, i)$
		\EndFor
		\label{line:reveal cm to m end}
		\State \textbf{return} (\textit{seed}, $Sk_\textit{TL}, L_s, L_m$)
		\EndFunction
	\end{algorithmic}		
	\caption{Coordinator Reveal}
	\label{alg:coordinator reveal}
\end{algorithm}

\begin{algorithm}[t]
	\begin{algorithmic}[1]
		\Function{VerifyOutput} {$\lambdam, \textit{seed}, Sk_\textit{TL}, L_s, L_m$}
		\State $Pk_\textit{TL} \gets \text{GenPuz}(\lambdam, \textit{seed}) $
		\label{line:verify pk from seed}
		\If {$Sk_\textit{TL}$ is not the solution to puzzle $Pk_\textit{TL}$}                     
		\label{line:verify_sk}
			\State \textbf{return} FALSE
		\EndIf
		\ForAll {$j \in \{1,...,N\}$}
		\label{line:verify lm to Lh start}
			\State $(m, i) \gets L_m[j]$
			\State $C_{m} \gets$ enc\textsubscript{$Pk_\textit{TL}$}($m$)
			\State $L_C[j] \gets (C_m,i)$
			\label{line:verify create lc}
			\State $H_{C_{m}} \gets \fun{hash}(C_{m})$
			\label{line:verify hm from cm}
			\State $L_H[j] \gets (H_{C_{m}},i)$
			\label{line:verify create lh}
		\EndFor
		\label{line:verify lm to Lh end}
		\ForAll {$j \in \{1,...,N\}$}
			\State $(C_m, i) \gets L_C[j]$
			\State $Pk_i \gets \fun{get\_party\_pk}(i)$
			\If {not {\fun{verifySig}($Pk_i$, $L_H\parallel\textit{seed}$, $L_s[j]$)}}
			\label{line:verify_sig}
				\State \textbf{return} FALSE
			\EndIf
		\EndFor
		\State \textbf{return} TRUE
		\EndFunction				
	\end{algorithmic}		
	\caption{VerifyOutput}
	\label{alg:verify_reveal}
\end{algorithm} 

\begin{algorithm}[t]
	\begin{algorithmic}[1]
		\Function {SolveDL}{p, g, b}
		\ForAll {i in [1, p-1]}
		\If {$g^i \text{ mod }p = b$}
		\State \textbf{return} i
		\EndIf
		\EndFor
		\EndFunction
	\end{algorithmic}		
	\caption{find ElGamal $Sk$ na\text{\"i}ve implementation}
	\label{alg:find_elgamaml_sk}
\end{algorithm}

	\section{Correctness} 
    \label{sec:correctness} 

 We prove that \alg solves \problem by proving that it achieves Soundness~(\S\ref{sec:correctness soundness}), FR-resistance~(\S\ref{sec:correctness fr}), Time-Lock~(\S\ref{sec:correctness time-lock}) and Ensured Output~(\S\ref{sec:correctness ensured output}) .
  
		\subsection{Soundness} 
		\label{sec:correctness soundness}

We prove that an adversary \Ad cannot output a valid $L_m$ with a message that wasn't committed in the commit phase.

\begin{lemma}
	For any PPT adversary \Ad playing \gamee{Sound}{\alg}, the advantage $\text{Sound}\textsf{adv}[\lambdah, \Ad]$ is negligible in $\lambdah$.
	\label{lemma:algorithm soundness}
\end{lemma}

\begin{proof}
	We calculate the probability that \Ad outputs a valid output with a list $L_{m_o}$ containing a tuple $m_o^l \in \{\{m_i^j \in L_{m_o}|P_j \in P_\textit{challenger}\}\} \backslash \{\{m_1^1,....,m_k^k\}\}$.
	Since the output is valid, the function \fun{VerifyOutput}($\textit{seed}, Sk_\textit{TL}, L_s, L_{m_o}$) (algorithm \ref{alg:verify_reveal}) returns true, and specifically in line \ref{line:verify_sig} the verification doesn't fail. 
	In this line the verifier checks that challenger's signatures over the list of commitments and seed are valid.
	We denote by $L_{H_o}$ the list of commitments that is calculated from the list $L_{m_o}$ by running lines \ref{line:verify lm to Lh start}-\ref{line:verify lm to Lh end} in function \fun{VerifyOutput}($\textit{seed}, Sk_\textit{TL}, L_s, L_{m_o}$).
	
	To pass verification, one of the following options must hold:
	\begin{enumerate}
		\item $1\leq l \leq k \wedge C_{m_o}=C_{m_l}$, The ciphertext of $m_o$ equals $P_l$'s committed ciphertext.
		\label{sound different cipher}
		\item $1\leq l \leq k \wedge C_{m_o} \neq C_{m_l} \wedge \fun{hash}(C_{m_o})=\fun{hash}(C_{m_l})$,  The ciphertext of $m_o$ is different from $P_l$'s, but their hashes are the same.
		\label{sound different ro}
		\item $L_{H_o} \neq L_H \wedge (\forall 1\leq j \leq k : \text{sign}_{Sk_j}(L_{H_o} \parallel \textit{seed}) \in L_s)$, the commitments lists are different but the challenger's signatures appears in the output's list of signatures $L_s$.
		\label{sound different sig}
	\end{enumerate}

	Let us start from option \ref{sound different cipher}. ElGamal encryption is an injective function, so if the ciphertexts are equal and the encryption keys are equal then the messages are equal, though by definition $m_o \neq m_j$.
	If the encryption keys are different, it means that each key was generated from a different seed value, $\textit{seed}_o \neq \textit{seed}$ (Algorithm \ref{alg:verify_reveal} line \ref{line:verify pk from seed}). In \fun{VerifyOutput}($\cdot$) line \ref{line:verify_sig} the signature over $L_H \parallel \textit{seed}$ is verified so in order to pass this check using a different seed \Ad has to forge the signature over $L_{H_o} \parallel \textit{seed}_o$, which she can do with probability $\varepsilon_{\textit{sig\_forgery}}$.
	The probability of option \ref{sound different ro} is the probability to find a hash collision, which is $\varepsilon_{\textit{hash\_collision}}$.
	Finally, for option \ref{sound different sig} the adversary has to create signatures on a new list of commitments $L_{H_o}$, that is, accomplish signature forgery given a polynomial computation limitation; the probability for that is at most $\varepsilon_{\textit{sig\_forgery}}$.

	In summary, the probability that \Ad can output $m_o^l$ in a valid output is
	$Pr[W^{\textit{\alg}}_{\textit{Sound}}] \leq \varepsilon_{\textit{hash\_collision}} + 2\cdot \varepsilon_{\textit{sig\_forgery}}$,
	which is negligible in $\lambdah$. Thus, \alg provides Soundness.
\end{proof}

\subsection{Frontrunning}
\label{sec:correctness fr}
Next, we prove that \alg is FR-resistant.
The FR attack is formalized using two games.
We first show that \alg is secure against Frontrunning Order Choosing (FOC).

\begin{lemma}
	For any PPT adversary \Ad playing \gamee{FOC}{\alg}, the advantage $\text{FOC}\textsf{adv}[\lambdam, \lambdah, \Ad]$ is negligible in $\lambdah$.
	\label{lemma:FOC}
\end{lemma}

\begin{proof}
	We denote by~$L_\textit{final}$ the tuple of transactions the attacker outputs and by~$\mathcal{L}_\textit{win} = \{(\text{valid }L|\{\{m_0^i \text{ appears in $L$ before } m_A\}$ or~$\{m_A \text{ appears in $L$ before } m_1^i\}\}$ and~$\{m_0^i \text{ and } m_1^i \text{ do not appear together}\})\}$ the set of winning lists.
	We calculate the probability that \Ad outputs $L_\textit{final} \in L_\textit{win}$.
	We can separate the calculation of the probability to win to two different paths. The first is that~\Ad created~$L_H$ (Figure \ref{commit protocol} line \ref{line:gen_list}) that is already organized in a winning order; we denote this event by~$E_{L_H^\textit{win}}$. The second option is that~\Ad created~$L_H$ that is not a winner, we denote this event by~$E_{L_H^\textit{lose}}$, but still won the game.
	\begin{equation}
	\label{eq:FR win}
	Pr[W]=Pr[E_{L_H^\textit{win}}]+Pr[E_{L_H^\textit{lose}}]\cdot Pr[L_\textit{final} \in \mathcal{L}_\textit{win}| E_{L_H^\textit{lose}}]
	\end{equation}
	
	We begin by calculating~$Pr[E_{L_H^\textit{win}}]$. When~\Ad chooses the order of~$L_H$, she only knows the value of the hash~$\fun{hash}(C_{m_b})$ and doesn't know the values~$C_{m_b}$ and~$m_b$.
	Since~$C_{m_b}$ is the encryption of~$m_b \| r$ and~$r$ is chosen uniformly at random from~$\{0,1\}^{\lambdah}$, the probability to find the pre-image of the hash function is~$\varepsilon_{\textit{pre-image}}$.
	If~$Pr[E_{L_H^\textit{win}}] > \frac{1}{2}$ is means that~\Ad has information about the message~$m_b$ at this stage. If she has information with probability larger than~$\varepsilon_{\textit{pre-image}}$ then she can also find the pre-image of the hash function with probability larger than~$\varepsilon_{\textit{pre-image}}$, violating the pre-image resistance assumption. So
	\begin{equation}
	\label{eq:FR LH win}
	\frac{1}{2} \leq Pr[E_{L_H^\textit{win}}] \leq \frac{1}{2} + \varepsilon_{\textit{pre-image}}
	\end{equation}
	and therefore, 
	\begin{equation}
	\label{eq:FR LH lose}
	Pr[E_{L_H^\textit{lose}}] = 1 - Pr[E_{L_H^\textit{win}}] \leq \frac{1}{2}.
	\end{equation}
	
	In the event $E_{L_H^\textit{lose}}$ \Ad can still output $L_\textit{final} \in \mathcal{L}_\textit{win}$. In order to do so, the output must be valid according to the function \fun{VerifyOutput}($\textit{seed}, Sk_\textit{TL}, L_s, L_m$). This function returns false if the challenger's signature over $(L_H \parallel \textit{seed})$ is not valid (algorithm \ref{alg:verify_reveal} line \ref{line:verify_sig}). 
	If the calculated list $L_H$ is different from the output then \Ad can try to forge the challenger's signature with probability $\varepsilon_{\textit{sig\_forgery}}$.
	Meaning, 
	\begin{equation}
	\label{eq:FR valid}
		Pr[L_\textit{final} \in \mathcal{L}_\textit{win}| E_{L_H^\textit{lose}}] \leq \varepsilon_{\textit{sig\_forgery}}.
	\end{equation}

	By plugging inequalities \ref{eq:FR LH win}, \ref{eq:FR LH lose} and \ref{eq:FR valid} into equation \ref{eq:FR win}, we bound the total probability to win game $\gamee{FOC}{\alg}$,
	\begin{eqnarray}
	\label{eq:FR total}
	Pr[W^{\textit{\alg}}_{\textit{FOC}}] & \leq &
	\frac{1}{2} + \varepsilon_{\textit{pre\_image}} + \frac{1}{2}\cdot\varepsilon_{\textit{sig\_forgery}}.\nonumber
	\end{eqnarray}
	Thus, $Pr[W^{\textit{\alg}}_{\textit{FOC}}]$ is negligibly larger than $\frac{1}{2}$ and  $\text{FOC}\textsf{adv}[\lambdam, \lambdah, \Ad]$ is negligible in $\lambdah$.
\end{proof}
	
\alg is also secure against Frontrunning Message Choosing (FMC), the proof of FMC is similar to the proof for FOC attack.

\begin{lemma}
		For any PPT adversary \Ad playing \gamee{FMC}{\alg}, the value $\text{FMC}\textsf{adv}[\lambdam, \lambdah, \Ad]$ is negligible in $\lambdah$.
	\label{lemma:FMC}
\end{lemma}

In summary, following directly from lemmas \ref{lemma:FOC} and  \ref{lemma:FMC}:

\begin{corollary}
	\alg provides FR-resistance.
	\label{corollary:algorithm fr resistance}
\end{corollary}

    		\subsection{Time-Lock}
  		\label{sec:correctness time-lock}

To show that \alg provides Time-Lock, we first prove that each puzzle generated by \alg is as hard as a DL puzzle (\S\ref{sec:correctness genpuz}).
Then we prove that using multiple different seeds for the puzzle creation does not give more than a negligible advantage in solving the puzzle (\S\ref{sec:correctness multi seed}).
And finally, we conclude that \alg provides Time-Lock (\S\ref{sec:time lock conclusion}).

  		\subsubsection{Puzzle Hardness} \label{sec:correctness genpuz}
  
  We first prove that a single puzzle generated by $\fun{GenPuz}(\lambdam, \textit{seed})$ is as hard as a random DL puzzle, except for a negligible probability. 

  To do this formally we first introduce some notation. $\mathcal{P}(\lambdam)$ is the set of all safe primes of size~$\lambdam$ bits; $\mathcal{G}(p)$ is the set of all generators of $Z^{\ast}_p$ and $\mathcal{B}(p)$ is the set of all members of $Z^{\ast}_p$: ($[a,b]$ is the range of integers from $a$ to $b$, inclusive).
  We define the following game \gamee{DLGenPuz}{\alg} of an adversary trying to solve our generated puzzle:
  
  \mytextbox{
  \begin{itemize}
  	\item The challenger chooses $seed \randselect \{0,1\}^{\lambdah}$.
  	\item The challenger computes $(p, g, b) \gets \fun{GenPuz}(\lambdam, seed)$ and gives it to the adversary.
  	\item The adversary outputs $a \in Z^{\ast}_p$.
  \end{itemize}
  
  Let $W_\textit{GenPuz}$ be the event that $b=g^a \text{ mod } p$ in  \gamee{DLGenPuz}{\alg}. Let $W_\textit{DL}$ be the event that the adversary succeeds to solve a DL puzzle with the same $p$ and $g$ and a random $b$.
  We show that a PPT adversary \Ad has negligible advantage in winning \gamee{DLGenPuz}{\alg} compared to a random DL puzzle. Therefore we define the adversary's advantage in the game as
  \begin{equation}
  \text{DLGenPuz}\textsf{adv}(\mathcal{A}, \lambdam)\triangleq \left|Pr[W_\textit{GenPuz}]- Pr[W_\textit{DL}]\right|.\nonumber
  \end{equation}
}

  \begin{lemma}
  	Let $\lambdam$ be a security parameter. For any PPT adversary \Ad, the advantage $\text{DLGenPuz}\textsf{adv}(\mathcal{A}, \lambdam)$ is negligible in $\lambdah$.
  	\label{lemma:hard_puzzle}
  \end{lemma}

  We first show the following claim:
  

  \begin{claim}
  	$\fun{GenPuz}(\lambdam, seed)$ is a pseudo random generator (PRG).
  	\label{lemma_prg}
  \end{claim}
  
  $\fun{GenPuz}(\lambdam, seed)$ has two logical elements. One element takes the random seed value and expands it to a longer pseudo random stream using a PRNG. The second element takes the expanded pseudo random stream and uses it to produce a pseudo random $b$. 
  
  Since we assume that the PRNG is secure for a polynomial length output, it remains to prove that a polynomial length pseudo random bit stream is sufficient.
  
  The function $\fun{GenPuz}(\lambdam, seed)$ returns deterministic values for $p$ and $g$ per $\lambdam$. The value of $b$ is chosen using iteration sampling from the PRNG output, trying the next $\lambdam$ bits in each iteration. The value of $b$ must be in the range $[2,p-1]$ and $p$ is $\lambdam$ bits long with the most significant bit set to one, so the probability to get a number in the correct range in each iteration is more than $1/2$. Thus, the probability that $\fun{GenPuz}(\cdot)$ would need $k$ iterations drops exponentially with $k$ and a polynomial length output of PRNG will be sufficient for the generation of $b$.
  
  The PRG $\fun{GenPuz}(\lambdam, seed)$ inherits its advantage from the underlying PRNG and has advantage $\varepsilon_{PRNG}$.
  
  Having shown that $\fun{GenPuz}(\lambdam, seed)$ provides a pseudo random $b$, it follows that when we use this parameter as a part of a DL puzzle, the puzzle is almost as hard as one from truly a random source.
  
  We defer the details of Lemma \ref{lemma:hard_puzzle}'s proof to Appendix \ref{appsec:puzzle hardness proof}.
  
	\subsubsection{Multiple Seed Advantage} \label{sec:correctness multi seed}
A coordinator can create many different \textit{seed} values by changing a nonce (if she is also a participant) or by changing the order of nonces in the Merkle tree (Figure \ref{commit protocol} line \refc{line:gen merkle tree}).
If she has a set of preferred puzzles that she can solve efficiently, she can try to receive them through generation of different seed values. 
We prove that a PPT adversarial coordinator using \alg cannot gain an advantage compared to the DL problem from generating a polynomial number of seeds, except with a negligible probability.

We formalize this requirement with the game \gamee{DLMulti-Seed}{\alg}:

\mytextbox{
\begin{itemize}
	\item The challenger chooses a polynomial set of random seeds $\{s_1,...,s_k|s_i \randselect \{0,1\}^{\lambdah}\}$ and sends them to the adversary.
	\item The adversary outputs value $a$.
\end{itemize}

We say that the adversary wins the \gamee{DLMulti-Seed}{\alg} game if she outputs~$a$ such that~$\exists i : b_i=g_i^a \text{ mod } p_i, (p_i, g_i, b_i) \gets \fun{GenPuz}(\lambdam, s_i)$.
We denote this event by $W_{DLMulti-Seed}$.
The adversary's advantage in the game is
\begin{equation}
\text{DLMulti-Seed}\textsf{adv}[\lambdam, \mathcal{A}]\triangleq \left|Pr[W_\textit{DLMulti-Seed}]-Pr[W_\textit{DL}]\right|.\nonumber
\end{equation}
}

\begin{lemma}
	For any PPT adversary \Ad playing the \gamee{DLMulti-Seed}{\alg} game the quantity $\text{DLMulti-Seed}\textsf{adv}[\lambdam, \mathcal{A}]$ is negligible in $\lambdam$.
	\label{lemma:alg hard as DL}
\end{lemma}

\begin{proof}
	The number of possible puzzles per~$\lambdam$ is the size of the group~$\mathcal{B}(p)$ that~$b$ is chosen from. The minimal number is~$|\mathcal{B}(p)| = p_\textit{min}-2 = 2^{\lambdam-1}-2$. 
	The probability to generate a specific puzzle is~$\frac{1}{2^{\lambdam-1}-2}$. 
	Under the assumption that~\Ad can store a polynomial number in~$\lambdam$ of computations, the probability to generate a puzzle that~\Ad has in her memory is smaller than~$k \cdot \frac{poly(\lambdam)}{2^{\lambdam-1}-2}$, which we denote~$\varepsilon_{\textit{multi-seed}}$.
	This probability is negligible in~$\lambdam$.
	
	We proved that the advantage that~\Ad has in solving a puzzle generated by \fun{GenPuz}($\cdot$) function is~$\varepsilon_{PRNG}$.
	The total advantage~$\text{DLMulti-Seed}\textsf{adv}[\lambdam, \mathcal{A}]$ is the probability to generate a pre-computed stored puzzle and if this strategy fails, to solve the puzzle, taking into account the PRG advantage.
	Since we assume that~\Ad does not have advantage in parallel computation of many puzzles over a computation of a single puzzle (Assumption \ref{def:multi dl assumption}), her advantage is
	$\text{DLMulti-Seed}\textsf{adv}[\lambdam, \mathcal{A}] \leq \varepsilon_{\textit{multi-seed}} + \varepsilon_{\textit{PRNG}}$,
	which is negligible in $\lambdam$.
\end{proof}

		\subsubsection{Timed reveal}\label{sec:time lock conclusion}
\begin{lemma}
	\alg provides Time-Lock.
	\label{lemma:algorithm time lock}
\end{lemma}		
	
\begin{proof}
Lemma \ref{lemma:alg hard as DL} shows that a PPT adversary has only a negligible advantage solving the puzzle in \alg compared to solving a single DL puzzle, given a random seed value.

The seed value in \alg is pseudo random since each participant gives a random nonce and the coordinator gathers all nonces to a Merkle tree structure to generate a seed for \fun{GenPuz}($\lambdam, \textit{seed}$).
Given that at least one participant is following the protocol, her \textit{nonce} is truly random (Figure \ref{commit protocol} line p\ref{line:gen nonce}) and she verifies that this \textit{nonce} is a part of the seed (Figure \ref{commit protocol} line p\ref{line:verify_mp}). Thus and due to the Merkle tree properties, the generated seed is pseudo random.
 
The requirement that $\tau$ is the mean time it takes to solve a DL puzzle defines the value of~$\lambdam$. 
A coordinator can begin to solve the puzzle once she has the \textit{seed} value~(line \refc{line:seed is ready}), and starting from this moment it will take her a mean value of~$\tau$ to solve the puzzle.
Hence, \alg respects definition \ref{def:pmptc time lock} and provides Time-Lock.
\end{proof}

Note that in the the proof for Time-Lock we haven't used the number of participants~$N$, thus
\begin{corollary}
	\alg's mean reveal time $\tau$ is independent of the participant count $N$.
\end{corollary}

	\paragraph{Practical reveal time}
The DL problem is an active research area \cite{adrian2015imperfect, corrigan2018discrete, kleinjung2017computation, joux2013new, commeine2006algorithm}, so we are unable to calculate the reveal time distribution. There are many known algorithms, each algorithm has a different CDF of the time it takes to solve it. In practice, Index calculus \cite{joux2013new} and Number Field Sieve \cite{commeine2006algorithm} have a few precomputation steps, involving some randomness, until finally the discrete logarithm is calculated. So the time distribution is actually closer to the mean value.

		\subsection{Ensured Output}\label{sec:correctness ensured output}
We show that the output of a coordinator that is following the protocol is valid and then that the mean time it takes her to calculate the output is~$\tau$.
		
\begin{lemma}
	The output of a coordinator following the protocol \alg is valid.
\label{lemma:ensured output valid}
\end{lemma}	

\begin{proof}
	We will go over the lines of \fun{VerifyOutput}($\cdot$) (Algorithm~\ref{alg:verify_reveal}) and check where the validation of the output can fail.
	In line \ref{line:verify pk from seed} the function calculates the value~$Pk$ using $\fun{GenPuz}(\lambdam, \textit{seed})$, this is done exactly as the coordinator does it in the commit protocol (Figure \ref{commit protocol} line \refc{line:gen puzzle}). In line \ref{line:verify_sk} the DL puzzle solution~$Sk_\textit{TL}$ is validated. Since the coordinator solves the puzzle correctly~(Algorithm \ref{alg:coordinator reveal} line \ref{line:find sk}), the solution is valid.
	Then the function goes over the list~$L_m$, it encrypts every message~$m$ in the list and creates a list~$L_C$ of tuples~$(C_m,i)$ ordered in the same order as~$L_m$.
	In the reveal phase, the coordinator does the exact opposite calculation from the list~$L_C$ to the list~$L_m$ (Algorithm \ref{alg:coordinator reveal} lines \ref{line:reveal cm to m start}-\ref{line:reveal cm to m end}). Therefore the list~$L_C$ is the same in the verification and in the reveal phase.
	
	The function calculates the hash~$H_{C_m}$ of the ciphertext~$C_m$ and generates a list~$L_H$ of tuples~$(H_{C_m}, i)$ ordered in the same order as~$L_m$ (Algorithm \ref{alg:verify_reveal} lines \ref{line:verify hm from cm}-\ref{line:verify create lh}).
	In the commit phase, the coordinator verifies that~$P_i$ had sent~$C_m$ that is a pre-image of~$H_{C_m}$ (Figure \ref{commit protocol} line \refc{line:verify p cipher}), so the list~$L_H$ is the same in the verification and in the commit phase.
	
	Once it has the list~$L_H$, it verifies that each signature in~$L_s[j]$ is valid as a signature of the participant that committed the message in slot~$j$ of the list over the content~$L_H \parallel \textit{seed}$.
	The coordinator also verifies that this signature is valid in the commit phase (Figure \ref{commit protocol} \cref{line:verify p sig}).
	Therefore, this check must pass.
	
	We  proved that if the coordinator follows the protocol then the output is valid.
\end{proof}

\begin{lemma}
	The mean time it takes for a coordinator following the protocol \alg to calculate the output is $\tau$.
\label{lemma:ensured output time}
\end{lemma}	

\begin{proof}
	A coordinator following the protocol generates the puzzle using \fun{GenPuz}($\lambdam$, \textit{seed}) (Figure \ref{commit protocol} \cref{alg:generate_pk}).
	Due to lemma \ref{lemma:alg hard as DL}, the puzzle is as hard as a DL puzzle, except for a negligible advantage. Since $\lambdam$ value is set to determine that solving the puzzle takes mean value $\tau$ then this is the time the it takes for the coordinator to solve the puzzle and calculate the output.
\end{proof}

\begin{corollary}
	\alg provides Ensured Output.
\label{corollary:algorithm ensured output}
\end{corollary}
The proof follows directly from lemma \ref{lemma:ensured output time} and lemma \ref{lemma:ensured output valid}.

	\subsection{\alg Correctness}
In summary, following from lemma \ref{lemma:algorithm soundness}, corollary \ref{corollary:algorithm fr resistance}, lemma \ref{lemma:algorithm time lock} and corollary \ref{corollary:algorithm ensured output}:
\begin{theorem}
	\alg algorithm solves \problem.
\end{theorem}

	\section{Precomputation Resistant DL} \label{sec:pcrdl} 

The time it takes to reveal the commitments with \alg is the time it takes to compute a DL puzzle solution.
The mean time to reveal the commitments depends on which algorithm is used for solving DL.
Progress in DL analysis~\cite{adrian2015imperfect} has shown that precomputation can reduce solution time. 
This does not violate the correctness of~\alg, since the difficulty can be increased to reach the desired~$\tau$. 
However, an accurate choice of~$\lambdam$ requires knowing whether the adversary has performed precomputation. 

We note, however, that in all algorithms we are aware of, including state of the art solutions like Index Calculus~\cite{adleman1979subexponential} and Number Field Sieve~\cite{adrian2015imperfect}, precomputation is only effective for a specific group and a specific generator. 

If this is the case in general, we can prevent precomputation by selecting both the group~$G$ and the generator~$g$ randomly when generating a puzzle. 
We make the assumption formal (Section~\ref{sec:pcrdl model}), describe the \pcrdlalg implementation (Section~\ref{sec:pcrdl implementation}) and prove that it is precomputation-resistant (Section~\ref{sec:pcrdl correctness}).

		\subsection{\pcrdl Model}\label{sec:pcrdl model}

A \emph{\pcrdl generator} is a function GenPuz. This function takes as input a security parameter $\lambdam$ and a random value \textit{seed} and produces as output a DL puzzle $(p,g,b)$.
		
We assume that precomputation of discrete logarithm is possible for a specific cyclic group $Z^{\ast}_p$ and a generator $g$.
We call this information a \emph{hint}.
More formally,
\begin{definition}[Precomputation Storage]
A PPT adversary with memory $M$ can store hints that can accelerate the computation of up to a polynomial number in $\lambdam$ of puzzles. 
\end{definition}
This assumption is conservative compared to the memory consumption required by Number Field Sieve, which is super polynomial in $\lambdam$ per hint \cite{adrian2015imperfect}.

We use the empirical assumption \cite{shoup2009computational} that the number of safe primes in interval $x$ is $2C_2\frac{x}{\log^2(x)}$, where $C_2 \approx 0.66$ is the \emph{twin prime constant} .

We require \emph{Precomputation-Resistance}. A coordinator using precomputation techniques to store hints in her memory does not have advantage over a coordinator that doesn't have memory to store hints.
We define this requirement with the following precomputation-resistance game \gamee{PC}{GenPuz}.

\mytextbox{
\begin{itemize}
	\item The challenger chooses a polynomial set of seeds $\{s_1,...,s_k|s_i \randselect \{0,1\}^{\lambdah}\}$ and sends them to the adversary.
	\item  The adversary outputs $a \in Z^{\ast}_P$.
\end{itemize}
We say that the adversary solves the puzzle if she outputs~$a$ such that~$\exists i : b_i=g_i^a \text{ mod } p_i, (p_i, g_i, b_i) \gets \fun{GenPuz}(\lambdam, s_i)$.
We denote by~$W_M$ the probability of an adversary with memory~$M$ and unbounded time for precomputations to solve the puzzle.
We denote by~$W_\textit{Mreset}^\textit{opt}$ the probability of an adversary using the best algorithm with memory~$M$ that is reset at the beginning of the game to solve the puzzle.
We define the advantage in this game by $\text{PC}\mathsf{adv}[\lambdam, \Ad]=|Pr[W_\textit{M}]-Pr[W_\textit{Mreset}^\textit{opt}]|$.
}

\begin{definition}[DL-PC Resistance]
	Let $\lambdam$ be a security parameter. The DL puzzle generated by $\textit{puz} \gets \text{GenPuz}(\lambdam, \textit{seed)}$ is a \pcrdl puzzle if for any PPT adversary \Ad with memory $M$ polynomial in $\lambdam$, the advantage  $\text{PC}\mathsf{adv}[\lambdam, \Ad]$ is negligible in $\lambdam$.
\label{def:pcrdl pc resistant}
\end{definition}

		\subsection{\pcrdlalg Implementation}\label{sec:pcrdl implementation}

We implement the function $\fun{PCRGenPuz}(\lambdam, \textit{seed})$ (algorithm \ref{alg:generate_pcrdl_pk}). This function takes as input the security parameter $\lambdam$ and a randomness \textit{seed} and outputs the three pseudo random parameters $(p, g, b)$ that represents the discrete logarithm puzzle.

As in  $\fun{GenPuz}(\lambdam, \textit{seed})$, the function initially expands the seed to a longer pseudo random bit stream with a pseudo random number generator (PRNG). But unlike \fun{GenPuz}($\cdot$), it chooses all three parameters $(p,g,b)$ using the pseudo random bit stream.

First (lines \ref{line:pcrgenpuz_p_search_start}-\ref{line:pcrgenpuz_p_search_end}), it chooses the prime number $p$. It constructs a $\lambdam$ bit prime candidate, $p_\textit{candidate}$ by setting the most significant bit to 1 and the other $\lambdam-1$ bits to the first $\lambdam-1$ bits from the pseudo random bit stream.
It checks if $p_\textit{candidate}$ is a safe prime with the probabilistic Miller-Rabin primality test \cite{rabin1980probabilistic} over two numbers, $p_\textit{candidate}$ and $\frac{p_\textit{candidate}-1}{2}$.
If the candidate is not a safe prime, it constructs a new candidate from the next $\lambdam-1$ pseudo random bits. It iterates until a it finds a safe prime.

Next, it searches for a generator of the group $Z^{\ast}_P$ (lines \ref{line:pcrgenpuz_g_search_start}-\ref{line:pcrgenpuz_g_search_end}). 
It constructs a candidate $g_\textit{candidate}$ from the next bits in the pseudo random bit stream and checks whether it is a generator using the standard algorithm for finding a generator in a cyclic group \cite{katz1996handbook}. 
It does this iteratively until it finds a generator.
Finally, it chooses the public key $b$ in a similar way from the group $Z^{\ast}_P$ (lines \ref{line:pcrgenpuz_b_search_start}-\ref{line:pcrgenpuz_b_search_end}).
The function returns the three chosen parameters $Pk_\textit{TL} \gets (p,g,b)$.

\begin{algorithm}[t]
	\begin{algorithmic}[1]
		\Function{PCRGenPuz}{$\lambdam, seed$}
		\State $j_\textit{bitnum} \gets 0$
		\State $p_\textit{candidate} \gets 0$
		\label{line:pcrgenpuz_p_search_start}
		\While {(not \fun{primality\_test}($p_\textit{candidate}$)) or (not \fun{primality\_test}($\frac{p_\textit{candidate}-1}{2})$)}
		\State $p_\textit{candidate} \gets 2^{\lambdam-1}+\fun{PRNG}(\textit{seed}, j_\textit{bitnum}, \lambdam-1)$
		\State $j_\textit{bitnum} \gets j_\textit{bitnum} + \lambdam-1$
		\EndWhile
		\label{line:pcrgenpuz_p_search_end}
		\State $g_\textit{candidate} \gets 0$
		\label{line:pcrgenpuz_g_search_start}
		\While {not \fun{is\_group\_generator}($p_\textit{candidate}, g_\textit{candidate}$)}
		\State $g_\textit{candidate} \gets \fun{PRNG}(\textit{seed}, j_\textit{bitnum}, \lambdam)$
		\State $j_\textit{bitnum} \gets j_\textit{bitnum} + \lambdam$
		\EndWhile
		\label{line:pcrgenpuz_g_search_end}
		\State $b_\textit{candidate} \gets 0$
		\label{line:pcrgenpuz_b_search_start}
		\While {$b_\textit{candidate} \notin [2,p-1]$}
		\State $b_\textit{candidate} \gets \fun{PRNG}(\textit{seed}, j_\textit{bitnumt}, \lambdam)$
		\State $j_\textit{bitnum} \gets j_\textit{bitnum} + \lambdam$
		\EndWhile
		\label{line:pcrgenpuz_b_search_end}
		\State \textbf{return} $(p_\textit{candidate}, g_\textit{candidate}, b_\textit{candidate})$
		\EndFunction
	\end{algorithmic}		
	\caption{Generate \pcrdl Puzzle}
	\label{alg:generate_pcrdl_pk}
\end{algorithm} 
 
 		\subsection{\pcrdlalg Security} \label{sec:pcrdl correctness}

 Similarly to $\fun{GenPuz}(\lambdam, \textit{seed})$, we first prove that $\fun{PCRGenPuz}(\lambdam, \textit{seed})$ is a PRG.
 It is easy to see that iteration sampling provides a uniform choice for~$(p, g, b)$ in their respective sets (cf. Appendix~\ref{appsec:iteration sampling}).
 It remains to show that a polynomial pseudo random stream is sufficient in order to select~$p$ and~$g$.

 \begin{lemma}
	The probability that the length of pseudo-random stream needed to generate a safe prime is more than polynomial in $\lambdam$ is negligible.
 	\label{lemma:prime_bits_poly_bounded}
 \end{lemma}
 
 We use the empirical assumption defining the number of safe primes in an interval $x$ to prove lemma \ref{lemma:prime_bits_poly_bounded}. For the mathematical details see the full proof in Appendix~\ref{appsec:proof prime_bits_poly_bounded}.

Since $p$ is a safe prime, the number of generators of group $Z^{\ast}_p$ is roughly $\frac{p}{2}$. 
Thus, it is trivial that finding a valid $g$ also requires a polynomial number of iterations. Therefore, \pcrdlalg is a PRG:
 \begin{lemma}
 	$\fun{PCRGenPuz}(\lambdam, \textit{seed})$ is a pseudo random generator (PRG) choosing three parameters, each one is pseudo random and chosen from its set of options,~$p \in \mathcal{P}(\lambdam), g \in \mathcal{G}(p)$ and~$b \in Z^{\ast}_p$.
 	\label{lemma:pcrdl_prg}
 \end{lemma}

\begin{proof}
	The PRNG with advantage~$\varepsilon_\textit{PRNG}$ provides enough pseudo random bits so that $p$ is chosen uniformly at random from~$\mathcal{P}(\lambdam)$ (Lemma \ref{lemma:prime_bits_poly_bounded}) and $g$ and $b$ are chosen uniformly at random from~$\mathcal{G}(p)$ and~$Z^{\ast}_p$.
	We calculate the advantage of \fun{PCRGenPuz}($\cdot$) over a random choice.
	We define two games: 
	In game 0, the challenger computes $p \randselect\mathcal{P}(\lambdam), g \randselect \mathcal{G}(p), b \randselect Z^{\ast}_p$.
	In game 1, the challenger computes $s \randselect \{0,1\}^{\lambdah} ,(p,g,b) \gets \fun{PCRGenPuz}(s, \lambdam)$.
	The challenger in both games sends $(p,g,b)$ to the adversary and the adversary outputs one bit. We denote the event that the adversary outputs 1 in game t by~$W_t$.
	If the adversary has an algorithm that outputs~$1$ with probability difference larger than $\varepsilon_{PRNG}$, then it would contradict that the PRNG has advantage $\varepsilon_{PRNG}$. Thus,
	\begin{eqnarray}
	\text{PCRGenPuz}\textsf{adv}[\lambdam, \Ad] \triangleq |Pr[W_1]-Pr[W_0]| < \varepsilon_{PRNG}.\nonumber
	\end{eqnarray}
	The advantage is negligible and thus \fun{PCRGenPuz}($\cdot$) is a PRG.
\end{proof}

We are now ready to prove the security of \fun{PCRGenPuz}($\cdot$).

\begin{theorem}
	The \pcrdlalg is DL-PC resistant.
\label{theorem:pcrdl pc resistant}
\end{theorem}

\begin{proof}
	Due to Lemma \ref{lemma:pcrdl_prg}, \fun{PCRGenPuz}($\cdot$) is a PRG with advantage $\varepsilon_\textit{PRNG}$ and due to Lemma \ref{lemma:alg hard as DL}, \fun{GenPuz}($\cdot$) generates a puzzle, which has at most a negligible advantage over a random DL puzzle.
	Under the assumption that precomputation of DL is possible for a specific $p$ and $g$, \fun{PCRGenPuz}($\cdot$) puzzle is at least as hard as \fun{GenPuz}($\cdot$) puzzle, since it chooses pseudo random $p$ and $g$ per puzzle.
	
	By Assumption, \Ad can perform a precomputation per group and generator and the number of hints that \Ad can store is $poly(\lambdam)$. 
	
	The number of possible groups is the number of safe primes in interval $[2^{\lambdam-1},2^{\lambdam}-1]$, namely~$|\mathcal{P}(\lambdam)| \approx {C_2}\frac{2^{\lambdam}}{{\lambdam}^2}$, \eqref{eq:k1}. 
	Each group has approximately $|\mathcal{G}(2^{\lambdam})| = \frac{2^{\lambdam}}{2}$ generators.
	We denote the fact that a puzzle \textit{puz} has a \textit{hint} in $M$ by $\textit{puz} \in M$.
	The probability of the adversary to generate such a puzzle is
	\begin{eqnarray}
	Pr[\fun{PCRGenPuz}(\lambdam, \textit{seed}) \in M]& \approx & \frac{|\textup{hints in $M$}|}{|\mathcal{P}(\lambdam)|\times|\mathcal{G}(2^{\lambdam})|}\nonumber \\ 
	= \frac{\text{poly}(\lambdam)}{{C_2}\left(\frac{2^{\lambdam}}{{\lambdam}^2}\right)\left(\frac{2^\lambdam}{2}\right)} 
	& = & \frac{2\lambdam^2 \text{poly}(\lambdam)}{C_2 2^{2\lambdam}},\nonumber
	\end{eqnarray}
	which is negligible in $\lambdam$. If \Ad generates a polynomial number of puzzles $k$, then the probability that any of them is precomputed is bounded by $Pr[\fun{PCRGenPuz}(\lambdam, \textit{seed}) \in M]$ multiplied by a polynomial, which is negligible. We denote this probability by $\varepsilon_{pc}$.
	
	An adversary playing \gamee{PC}{\pcrdlalg} with a memory that is reset when the game begins has probability $Pr[W_\textit{Mreset}^\textit{opt}]$ to solve one of the $k$ puzzles.
	An adversary playing with a memory available for precomputations can use the same strategies that an adversary with memory $Mreset$ can use and in addition she can also use her memory to accelerate computations with probability $\varepsilon_{pc}$. In summary, her probability to solve the puzzle is $Pr[W_M]\leq Pr[W_\textit{Mreset}^\textit{opt}]+\varepsilon_{pc}$.
	
	Thus, the advantage is bounded, $\text{PC}\mathsf{adv}[\lambdam, \Ad] \leq \varepsilon_{pc}$, where $\varepsilon_{pc}$ is negligible in $\lambdam$, so \pcrdlalg is a PCR-DL.
\end{proof}


	\section{\bname} \label{sec:bc}

A specific scenario where frontrunning-resistance is crucial is \emph{blockchain smart contracts}. 
In blockchains such as Ethereum~\cite{buterin2013ethereum}, Bitcoin~\cite{nakamoto2008bitcoin}, Libra~\cite{baudet2019state} and Tezos~\cite{goodman2014tezos}, users submit transactions to a shared public log. 
The transactions can be simple monetary transfers or steps in elaborate so-called smart contracts. 
The order of transactions determines their outcome. 
One example is a smart contract implementing an exchange \cite{idex11,uniswap,bancor,radarrelay,etherdelta}, where Alice publishes a \$$10$ buy order and Carol publishes a \$$7$ sell order for some token.
If Bob is able to frontrun and set the order of transactions, he can buy the token from carol, sell it to Alice and earn \$$3$.
Frontrunning bots with optimized network delays are widely deployed~\cite{daian2019flash}, exploiting knowledge of future trades to frontrun. 

Most prominent blockchains~\cite{nakamoto2008bitcoin, buterin2013ethereum, litecoin, monero, bitcoincash} achieve censorship-resistance with  PoW~\cite{dwork1992pricing,jakobsson1999proofs}. 
But regardless of their censorship-resistance, PoW blockchains are particularly susceptible to frontrunning: user transactions are published over a peer-to-peer network, allowing anyone to see all transactions and potentially inject her own to gain an advantage. 
Each transaction carries a fee as incentive for the miner who places it in a block. 
Frontrunning is therefore as simple as offering a higher fee for an injected transaction. 
Moreover, miners are able to arbitrarily add or remove transactions when they compose a new block. 

Our goal is therefore to design a PoW alternative that prevents frontrunning until the work is complete. 
We focus on a block generation model~(\S\ref{sec:bc:model}) since other design choices~\cite{nakamoto2008bitcoin,buterin2013ethereum,pass2017fruitchains,bagaria2019prism,sompolinsky2018phantom}, such as graph structure, are orthogonal.
Beyond frontrunning-resistance, the protocol should also maintain PoW properties, mainly allow for difficulty tuning and prevent an attacker from cheaply generating a series of blocks or multiple blocks with the same parents. 

We present the \emph{Capsule-Chain} protocol~(\S\ref{sec:bc:protocol}) that achieves these properties. 
It uses the Time-Capsule approach to achieve frontrunning-resistance. 
It uses \pcrdlalg to ensure miners have no advantage from precomputation, maintaining the openness of standard PoW. 
In \bname, the PoW is bound to the location of the block in the block graph and the content of the block, so the same work cannot be used to produce multiple blocks.
Hence, \bname uses the hash pointers of the parents and hashes of the transactions to produce the puzzle.
This is different from \alg that uses arbitrary nonces.
We prove~(\S\ref{sec:bc:correctness}) that in \bname the amount of computation needed to create blocks is linear in the number of blocks and that \bname provides Soundness, FR-resistance and Time-Lock as in \alg and Precomputation-resistance as in \pcrdlalg.

We overview important elements outside the single-block generation model~(\S\ref{sec:bc:discussion}).
Unlike most blockchains, in \bname users send transactions directly to a specific miner and do not spread them across the network; the miner does not check transactions for validity in advance, since they are encrypted during the block creation process.
For the same reason, \bname does not provide Ensured Output and we discuss how the miner can cope.
Finally we discuss user's dilemma if her transaction is not placed in a block, as she has no guarantees it was not revealed. 


		\subsection{\bname Model} \label{sec:bc:model}

Analogically to the \problem model, the participants are now called \emph{users} and the coordinator is called a \emph{miner}. But unlike \problem, none of the parties are trusted to follow the protocol.
In addition to the user's messages, called \emph{transactions} in this context, the miner places in each block a reference to the previous block/s, denoted \emph{metadata}. 

We are looking for a \fun{createBlock} protocol and a \fun{verifyBlock} function.
\fun{createBlock} takes as input the metadata and outputs a new block containing a list of the transactions and the metadata and a proof of its validity.
Anyone can verify that a block is valid using \fun{verifyBlock}. 
The function takes as input a block and returns TRUE if and only if the metadata, the transactions and their order match the metadata, the transactions and their order as committed by the users that participated in \fun{createBlock}.

Blockchain protocols serve as their own PKIs~\cite{nakamoto2008bitcoin,buterin2013ethereum}, where each user's ID is represented on the blockchain by a public key. 
Thus, we assume a PKI despite the openness of the system.

The requirement from a solution are (1)~\emph{Soundness}, (2)~\emph{FR-Resistance} and (3)~\emph{Time-Lock} as defined for \problem; (4)~\emph{PC Resistance}, as in \pcrdl; and a new requirement (5)~\emph{Linearity of Computation}, which we define as follows. 

The amount of computation required to create $k$ blocks should be linear in $k$.
We formalize the requirement with the \gamee{Linear}{\bname} game, parametrized by a positive integer $k$ polynomial in $\lambdam$.

\mytextbox{
\begin{itemize}
	\item The adversary expends computation power equivalent to solving $k$ DL puzzles, in expectation.
	\item The adversary returns Capsule-Chain's blocks.
\end{itemize}
The adversary wins the \gamee{Linear}{\bname} game if she returns $k+1$ different valid blocks. We denote the event that the adversary wins by $W_{\textit{Linear}}^{\textit{\bname}}$. We define the advantage in this game by $\text{Linear}\textsf{adv}[\lambdam, k, \Ad] \triangleq Pr[W_{\textit{Linear}}^{\textit{\bname}}]$
}

\begin{definition}
	A PoW system provides Linearity of Computation if for every $k$, polynomial in $\lambdam$, and for any PPT adversary~\Ad playing \gamee{Linear}{\bname}, the advantage $\text{Linear}\textbf{adv}[\lambdam, k, \Ad]$ is negligible in $\lambdam$.
\end{definition} 

This requirement implies that the content of the block is binding for both the transactions and the metadata since the miner cannot create two different blocks without solving two different puzzles.


		\subsection{\bname Protocol} \label{sec:bc:protocol} 
		
The commit protocol of \bname (Appendix \ref{appsec:bc commit protocol}) is similar to that of \alg (Figure~\ref{commit protocol}) with the following differences:
\begin{inparaenum}
\item The first element of the Merkle tree (Figure \ref{commit protocol} line \refc{line:gen merkle tree}) is the hash of \textit{metadata}.
\item There is an additional output parameter, a list of the order of user's nonces in the Merkle tree.
\item The commit protocol uses PCRGenPuz($\cdot$) instead of GenPuz($\cdot$) (Figure~\ref{commit protocol} line~\refc{line:gen puzzle}) to generate a precomputation-resistant puzzle.
\item In the first step, each user calculates $nonce \gets \fun{hash}(m_i)$ instead of a random nonce (Figure \ref{commit protocol} line \refc{line:gen nonce}).
\end{inparaenum} 

In \alg, the problem becomes irrelevant if all parties are malicious, since the result doesn't affect any external party. In contrast, in \bname, a malicious party or colluding parties that create a block can affect other parties by violating the blockchain properties.
In Time-Capsule's commit protocol every participant sends a random nonce and the coordinator generates a pseudo random seed using a Merkle tree (Figure \ref{commit protocol} line \ref{line:gen merkle tree}), the seed is pseudo random since at least one nonce is truly random. 
This is essential for the proof that solving the puzzle generated from \textit{seed} has negligible advantage compared to a random puzzle.
In \bname, no party is trusted, so none of the nonces is necessarily random. 
The source of trusted randomness here is $\fun{hash}(\textit{metadata})$.

A miner creating a new block first runs the routine \fun{createBlock}(\textit{metadata}) (Algorithm~\ref{alg:miner_build_block}) with the metadata that she chooses.
This routine first runs the commit protocol, as described above.
Then, the miner runs the reveal function (Algorithm \ref{alg:coordinator reveal}), which is similar to the Time-Capsule reveal function. 		
In addition to the content required for the output of \alg, the block also contains \textit{metadata} and a list $L_\textit{seed}$ of users' indices ordered by the same order as their nonces in the merkle tree (Algorithm \ref{alg:miner_build_block} line \ref{line:return block}).
	
Anyone can verify the validity of the block using the function {VerifyBlock} (Algorithm \ref{alg:verify_block}). The function re-generates \textit{seed} by aggregating the hashes of the metadata and the hashes of the messages in a merkle tree according to $L_\textit{seed}$, i.e., in the same order that the tree was generated in the commit protocol. Using this \textit{seed}, it verifies the output using  \fun{VerifyOutput} as in \alg (algorithm \ref{alg:verify_reveal}).

\begin{algorithm}[t]
	\begin{algorithmic}[1]
	\Function{CreateBlock}{\textit{metadata}}
	\State $(\textit{seed}, T_\textit{seed}, Pk_\textit{TL}, L_c, L_s, L_\textit{seed}) \gets$ run commit protocol (Appendix \ref{appsec:bc commit protocol}).
	\State $(\textit{seed}, Sk_\textit{TL}, L_s, L_m) \gets \text{Reveal}(Pk_\textit{TL}, L_s, L_c)$  /*solve puzzle*/
	\State return $(\textit{metadata}, Sk_\textit{TL}, L_s, L_\textit{seed}, L_m)$
	\label{line:return block}
	\EndFunction
\end{algorithmic}
	\caption{Miner Build Block}
	\label{alg:miner_build_block}
\end{algorithm}

\begin{algorithm}[t]
	\begin{algorithmic}[1]
		\Function{VerifyBlock} {$\textit{metadata}, Sk_\textit{TL}, L_s, L_\textit{seed}, L_m$}
		\State Initiate a Merkle tree $T_\textit{seed}$
		\State Insert $\fun{hash}(\textit{metadata})$ to $T_\textit{seed}$
		\ForAll {$(m, i) \in L_m$}
			\State $j \gets$ get index of user $i$ in $L_\textit{seed}$
			\State Insert $\fun{hash}(m)$ to leaf number $j$ in $T_\textit{seed}$ 
		\EndFor
		\State $\textit{seed} \gets T_\textit{seed}$'s root
		
		\If {not \fun{VerifyOutput}($\textit{seed}, Sk_\textit{TL}, L_s, L_m$) (algorithm \ref{alg:verify_reveal})} 
			\textbf{return} FALSE
		\EndIf
		\State \textbf{return} TRUE
		\EndFunction				
	\end{algorithmic}

	\caption{	\tiny	Verify Block}
	\label{alg:verify_block}
\end{algorithm} 


		\subsection{\bname Correctness} \label{sec:bc:correctness} 


Both the Soundness and FR Resistance proofs for \alg (\S\ref{sec:correctness soundness} and \S\ref{sec:correctness fr}) apply for \bname. 
The proofs refer to the commit protocol stages after the seed creation. 
Since the differences in the commit protocol are only in the seed creation, they do not affect these proofs.
We proceed to overview the rest of the properties.

		\subsubsection{Time-Lock}

\bname provides Time-Lock.

The proof is similar to the proof for \alg (\S\ref{sec:correctness time-lock}) except for the argument for seed pseudorandomness, since other than the seed creation, the commit protocol is the same for both.
 
The source of trusted randomness in \bname is a hash of the state of the current block graph (\textit{metdata}).
A PPT adversary can try to run a brute force search on a block in order to get a desired $\fun{hash}(\textit{metadata})$, but her chances to succeed are negligible in $\lambdah$, due to the pre-image resistance of the hash.

		\subsubsection{PC Resistance} 

We proved (Theorem \ref{theorem:pcrdl pc resistant}) that the puzzle generated by PCRGenPuz($\cdot$), which \bname uses, is PC-resistant.
The difference in \bname is that the seed that is used as input for PCRGenPuz($\cdot$) is pseudo random and not random as in the proof of Theorem \ref{theorem:pcrdl pc resistant}. 
A pseudo random \textit{seed} input to a PRG function gives a negligible advantage over a random \textit{seed} input to this function, due to the PRG guarantees (Appendix \ref{appsec:puzzle hardness proof}).
Other than the \textit{seed}, the proof is the same as the proof for \pcrdlalg.
Thus, \bname provides PC resistance.

		\subsubsection{Linearity of Computation}

We prove that the puzzle computation effort is linear in the number of puzzles.
		
\begin{lemma}
	For every PPT adversary \Ad playing \gamee{Linear}{\bname}, and for every value $k$, polynomial in $\lambdam$, the advantage $\text{Linear}\textbf{adv}[k, \Ad]$ is negligible in $\lambdam$.
\label{lemma:bc linearity of computation}
\end{lemma}

\begin{proof}
	If \Ad publishes $k+1$ blocks, it means she has puzzle solutions for $k+1$ blocks.
	Due to the PC-resistance property of the puzzle,~\Ad has a negligible advantage~$\varepsilon_{pc}$ to use a common precomputation of different puzzles.
	There is only one other way to solve $k$ puzzles and create $k+1$ blocks, to have at least two identical puzzles in the set of $k+1$ puzzles.

	Consider two blocks, $B$ and $B'$.
	Block $B$ contains $l$ messages and \textit{metadata}; we denote the $i$'th message of block $B$ by $m_i$. 
	Similarly, block $B'$ contains $l'$ messages and \textit{metadata'}; $m_i'$ is message number $i$ in $B'$.
	We denote by \textit{puz} and \textit{puz'} the puzzles generated from block contents $B$ and $B'$, respectively. 
	
	There can be two sources for block differences, different metadata or different committed messages.
	Let us find the probability for the adversary to generate two different blocks $B$ and $B'$ with the same puzzle $\textit{puz}=\textit{puz'}$.
	
	For this to happen, one the following must occur:
	\begin{enumerate}
		\item The lists of messages are of the same length $l = l'$ and at least one message is different ($\exists i, 1\leq i \leq l$ such that $m_i \neq m_i'$) or the metadata is different ($\textit{metadata} \neq \textit{metadata}'$), but their hashes are the same ($\fun{hash}(\textit{metadata})=\fun{hash}(\textit{metadata'})$ and $\forall i$, $\fun{hash}(m_i)=\fun{hash}(m'_i)$);
		\label{item:diff content}
		\item at least one node of the Merkle tree is different (including the case $l \neq l'$) but the Merkle root is the same for both blocks; or
		\label{item:diff merkle nodes}
		\item the blocks' seed values are different but \fun{PCRGenPuz} produces the same puzzle.
		\label{item:pcrgenpuz collision}
	\end{enumerate}
	
	The probability for (\ref{item:diff content}) is $\varepsilon_{\textit{hash\_collision}}$ and the probability for (\ref{item:diff merkle nodes}) is $\varepsilon_{\textit{MT\_collision}}$.
	
	It remains to calculate the probability of~(\ref{item:pcrgenpuz collision}), which we denote~$\varepsilon_{\textit{PCRGenPuz\_collision}}$. 
	The number of possible puzzles per~$\lambdam$ is~${C_22^\lambdam}\cdot\left(\frac{2^{\lambdam}}{{\lambdam}^2}\right)\cdot\left(\frac{2^\lambdam}{2}\right)$ (Equation \ref{eq:k1}).
	So the probability to generate a puzzle \textit{puz} from \textit{seed} is roughly the probability to choose a puzzle uniformly at random from the set of puzzles,  $(\frac{C_22^{3\lambdam}}{2\lambdam^2})^{-1}$, thus the probability $\varepsilon_{\textit{PCRGenPuz\_collision}}$ is negligible in $\lambdam$.
	
	The total probability that block contents $B$ and $B'$ would produce the same puzzle is less than~$\varepsilon_{\textit{hash\_collision}} + \varepsilon_{\textit{MT\_collision}} + \varepsilon_{\textit{PCRGenPuz\_collision}}$.
	Both $\varepsilon_\textit{MT\_collision}$ and  $\varepsilon_\textit{hash\_collision}$ are negligible in $\lambdah$ and $\varepsilon_{\textit{PCRGenPuz\_collision}}$ is negligible in $\lambdam$.
	Thus the probability to obtain two blocks $B$ and $B'$ such that $\textit{puz'}=\textit{puz}$ is negligible in $\lambdam$.
	The advantage in the game is the probability to use common precomputations or to get an identical puzzle for two different blocks, multiplied by the number of possible block pairs in $k+1$ puzzles,
	$\text{Linear}\textbf{adv}[k, \Ad] \leq \varepsilon_{pc} + \frac{k^2}{2} \cdot Pr[\textit{puz}=\textit{puz}'|B\neq B']$.
	Since $k$ is polynomial in $\lambdam$, the probability is negligible in $\lambdam$. 
	Therefore, \bname provides Linearity of Computation.
	\end{proof}
	
	\subsection{Practical Considerations} \label{sec:bc:discussion}

We review several challenges for \bname that remain as open questions for future work. 

		\subsubsection{Mining Advatange} \label{sec:bc:discussion mining}
Solving the DL problem is a search in a finite set. 
The mining process in \bname is therefore not memoryless as with hash-based mining~\cite{nakamoto2008bitcoin}. 
Thus, larger miners gain an additional advantage, improving their chance to find a block beyond the proportion of their computational power. 
The implication is reduced motivation for small miners to participate, an issue that should be addressed with mechanism design. 

		\subsubsection{Ensured Output} \label{sec:bc:discussion ensured}
\bname does not provide Ensured Output as \alg since the hash that the user sends in the beginning of the commit protocol must be the same as the hash of the message in the output block in order for the block to be valid. 
The miner cannot validate this until she reveals the committed messages, and a violation by a user implies the loss of a full block's computation work by the miners. 
One way to overcome this is to have each user send a zero-knowledge proof that her commitment matches her nonce before forming the commitment list (Figure~\ref{commit protocol}).
Another is to incentivize users to follow the protocol by external measures, such as penalizing misbehavior, secured by a collateral, possibly on the blockchain itself. 
Note that this is easily enforceable as the miner has irrefutable proof for misbehavior in form of the violating user's signature ($S_{{L_H},i}$). 

		\subsubsection{Transaction Fee}\label{sec:bc:discussion fees}
As with existing PoW solutions, the miner is incentivized to create non-empty blocks due to transaction fees.
The transaction fee rate would naturally balance to cover the overhead of the interactive commit protocol.

		\subsubsection{Aborted Block}\label{sec:bc:discussion aborted}

In \bname, many miners compete to create the next block. 
Once a miner mined a block, she publishes it and then other miners switch to mine the subsequent block.

Consider a user that completed the commit phase with a miner~$A$, but~$A$ did not complete the block creation while another miner~$B$ did.
In principle, miner $A$ is supposed to abandon the transaction and start mining the next block. 
But the user doesn't have any guarantee that the miner indeed hasn't solved the puzzle before moving on. 
If the miner had found the puzzle solution, it means that all transactions in that block are not secret anymore.

The user now has a dilemma. 
If her intentions remain secret, she might want to post the same transaction again in the next block; but if her intentions were revealed, she might want to change strategy and place a different transaction or none at all.
The user might suspect that the miner has intentionally rejected this block despite expending the resources and solving the puzzle. 
Her decision depends on the application and is outside the scope of this work. 

However, we note that the chances for such a reveal are small. 
Typically \cite{nakamoto2008bitcoin, buterin2013ethereum, sompolinsky2018phantom}, once a block is published all miners should change the metadata of the block they are generating.
Since many miners compete to create each block, a small miner is not likely to win. 
A miner that owns $10\%$ of the computation power in the system would take in expectation roughly another 9 times the block interval to reveal a block. 
Thus, the miner in unlikely to quickly reveal the transactions, and would only reject a completed block in radical situations, i.e., the block's content will cause her a larger loss than the block's creation profit.


	\section{Conclusion} \label{sec:conclusion} 

We define \problem, where multiple parties commit to values that are revealed after a tunable expected time. 
It is useful in any situation where frontrunning is unacceptable. 
We present~\alg that solves \problem and prove its correctness without relying on trust, incentives or synchrony as previous approaches to similar problems. 
We also show how~\alg can be extended to prevent precomputation using assumptions based on existing solutions of the~DL problem in~$\mathbb{Z}_n^*$. 
In \bname we adapt \alg to overcome the frontrunning issue plaguing blockchain systems, resolving a persistent challenge in decentralized permissionless system design.

	\section*{Acknowledgements} \label{sec:acknowledgements}

We thank Vitaly Shmatikov and Yehuda Lindell for their feedback and suggestions.

\clearpage
\bibliographystyle{plain}
\bibliography{ref_list}
\appendix
\section{Iteration Sampling} \label{appsec:iteration sampling}
In this section we prove that a value chosen from a set $Z$ by iteration sampling from a larger set $Y$ is chosen uniformly at random.
\begin{definition}
	Let $Y$ be a set of values. Let $Z$ be set of values such that $Z \in Y$. We say that $x$ is chosen from $Z$ by iteration sampling from $Y$ if we iteratively choose a random value $r \randselect Y$ and check if $r \in Z$. If it is then $v \gets r$.
\end{definition}

\begin{lemma}
	Let $Y$ be a set of values. Let $Z$ be set of values such that $Z \in Y$. A value $x$ chosen from $Z$ by iteration sampling from $Y$ is chosen uniformly at random.
\end{lemma}

\begin{proof}
	The probability to choose a value $x \in Z$ at the first iteration is $Pr[\textup{success}]=\frac{|Z|}{|Y|}$ and the probability to choose a value $x \notin Z$ is $Pr[\textup{failure}]=\frac{|Y|-|Z|}{|Y|}$.
	The probability to choose $x \in Z$ at iteration k is $Pr[\textup{success at iteration k}]= Pr[\textup{failure}]^{k-1}\times Pr[\textup{success}]$.
	The probability to choose a specific $x$ value is
	\begin{eqnarray}
	\label{eq:1}
	Pr[x]  & = & \sum_{k=1}^{\infty} \left(Pr[\textup{success at iteration k}]\times \right.\nonumber\\ 
	& & \left. Pr[x|\textup{success at iteration k}]\right)\nonumber\\
	&= &\overbrace{\sum_{k=1}^{\infty} \left(\frac{|Y|-|Z|}{|Y|}\right)^{k-1}\times\frac{|Z|}{|Y|}}^{1}\times\frac{1}{|Z|}\nonumber\\
	&=&\frac{1}{|Z|}.
	\end{eqnarray}
	
	This result indicates that $x$ is chosen uniformaly at random from the set $Z$.
\end{proof}

\section{Puzzle Hardness Proof}\label{appsec:puzzle hardness proof}

We prove Lemma \ref{lemma:hard_puzzle}, given the \fun{GenPuz}($\cdot$) is a PRG as shown in Claim \ref{lemma_prg}.
\begin{proof}
	We prove that if \Ad had a non-negligible advantage in the game, then a PRG adversary could have used \Ad's output in order to distinguish between a PRG output and a truly random value. 
	
	Recall that PRG security is defined using two games \cite{boneh2017graduate}. Let $G$ be a PRG.
	In Game 0, $\mathcal{B}$ gets as input a truly random $b$, in Game 1, $\mathcal{B}$ gets as input $b$ such that $s \randselect \{0,1\}^{\lambdah}, b \gets G(s)$.
	$\mathcal{B}$ outputs 1 in Game 0 with probability $Pr[W_{\textit{PRG,0}}]$ and in Game 1 with probability $Pr[W_{\textit{PRG,1}}]$.
	
	According to PRG security definition, for every PPT adversary $\mathcal{B}$, the advantage is
	\begin{equation}
	\label{eq:2}
	\textit{PRGadv}[B,G] = |Pr[W_{\textit{PRG,1}}]-Pr[W_{\textit{PRG,0}}]| < \varepsilon_{\textit{PRG}}
	\end{equation}
	
	Figure \ref{proof_prg_pcrdl} illustrates the following scenario.
	Let \Ad be a PPT adversary that attacks the DL problem.
	Let $\mathcal{B}$ be a PPT adversary that attacks PRG (in our case the PRG is \fun{GenPuz}($\cdot$)).
	We define $\mathcal{B}$ to play the role of the challenger for adversary \Ad. $\mathcal{B}$ receives $b$ from the PRG challenger and passes it on to adversary \Ad. \Ad returns her puzzle solution to $\mathcal{B}$. If \Ad successfully solves the discrete logarithm, then $\mathcal{B}$ returns 1, otherwise she returns 0. 
	
	The probability that $\mathcal{B}$ outputs 1 in Game 1 is the probabilty that \Ad solves a pseudo ransom puzzle,
	\begin{equation}
	\label{eq:3}	 
	Pr[W_{\textit{PRG,1}}]= Pr[W_{DLGenPuz}]
	\end{equation}
	The probability that $\mathcal{B}$ outputs 1 in game 0 is that probability that \Ad solves a random puzzle,  
	\begin{equation}
	\label{eq:4}
	Pr[W_{\textit{PRG,0}}]= Pr[W_{DL}]
	\end{equation}
	Plugging equations (\ref{eq:2}), (\ref{eq:3}) and (\ref{eq:4}) yields
	\begin{equation}
	\text{DLGenPuz}\textsf{adv}[\mathcal{A}, \lambdam] = \textit{PRGadv}[B,G] < \varepsilon_{PRG}
	\end{equation}
	and since the PRG advantage of \fun{GenPuz}($\cdot$) is $\varepsilon_{PRNG}$, then
	\begin{equation}
	\text{DLGenPuz}\textsf{adv}[\mathcal{A}, \lambdam] < \varepsilon_{PRNG}
	\end{equation}
	which is negligible in $\lambdah$.
\end{proof}

\begin{figure}
	\centering
	\setstretch{0.65}
	\begin{tikzpicture}[node distance=2.5cm,auto,>=latex']
	\begin{scope}
	
	\node[draw, fit={(0,0) (1.6,3)}, inner sep=0pt, label=center:{\begin{tabular}{c}PRG\\ Challenger\\\\   \\   \\   \end{tabular}}] (A) {};
	\node[fit={(0,0) (1.6,1.5)},   yshift=1cm, inner sep=0pt] (A1) {};
	\node[draw, fit={(0,0) (2.75,3)},  xshift=2.6cm, inner sep=0pt, label=center:{\begin{tabular}{l}$\mathcal{B}$\\\\   \\    \\  \\   \\   \\  \end{tabular}}] (B) {};
	\node[fit={(0,0) (2.75,2)}, inner sep=0pt, yshift=0.5cm, xshift=2.6cm] (B1) {};
	\node[fit={(0,0) (2.75,1.5)}, inner sep=0pt, yshift=1cm, xshift=2.6cm] (B2) {};
	\node[fit={(0,0) (2.75,1.5)}, inner sep=0pt, yshift=0.5cm, xshift=2.6cm] (B3) {};
	
	\node[draw, fit={(0,0) (1,2)},  xshift=3.475cm, yshift=0.5cm, inner sep=0pt, label=center:{\begin{tabular}{l}$\mathcal{A}$\\\\   \\    \\  \end{tabular}}](D) {};
	\node[fit={(0,0) (2.9,2)},  xshift=5.35cm, yshift=0.5cm, inner sep=0pt](E) {};
	\node[fit={(0,0) (0.25,1.5)},  xshift=8cm, yshift=0.5cm, inner sep=0pt](F) {};
	
	\path[->] (A1) edge node {$b$} (B2);
	\path[->] (B1) edge node {$b$} (D);
	\path[->] (D) edge node {$a$} (E);
	\path[->] (B3) edge node {\begin{tabular}{l}if $b=g^a \text{ mod } p$\\
		$\quad\quad \textit{return } 1$\\else\\ $\quad\quad \textit{return } 0$  \end{tabular}} (F);
	
	\end{scope}
	\end{tikzpicture}
	\caption{PRG adversary $\mathcal{B}$}
	\label{proof_prg_pcrdl}
\end{figure}
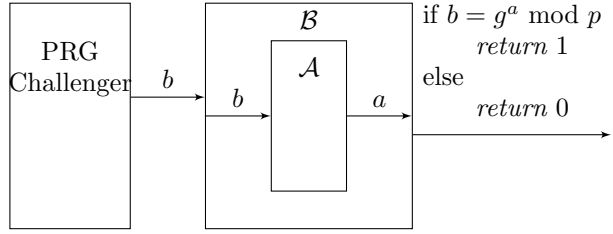

\section{Choosing A Safe Prime} \label{appsec:proof prime_bits_poly_bounded}

In this section we prove Lemma \ref{lemma:prime_bits_poly_bounded}.

\begin{proof}
	\fun{PCRGenPuz}($\cdot$) chooses a safe prime in the range  $[2^{\lambdam-1},2^{\lambdam}-1]$. 
	The number of safe primes in this range is 
	\begin{equation}
	\label{eq:k1}
	|\mathcal{P}(\lambdam)| = 2C_2\left(\frac{2^{\lambdam}}{\log^2(2^{\lambdam})}-\frac{2^{\lambdam-1}}{\log^2(2^{\lambdam-1})}\right) \cong {C_2}\frac{2^\lambdam}{\lambdam^2}
	\end{equation}
	The probability to fail to find a safe prime in $k$ guesses is 
	\begin{equation}
	Pr[\textup{k failures}]=\left(\frac{2^\lambdam-|\mathcal{P}(\lambdam)|}{2^\lambdam}\right)^k.\nonumber
	\end{equation}
	We show that 
	\begin{equation}
	\forall c \enspace \exists N : \forall \lambdam>N: Pr[\textup{k failures}] < \frac{1}{\lambdam^c}.\nonumber
	\end{equation}
	In other words, we are looking for $k$, such that the probability for $k$ failures is negligible.
	\begin{equation}
	\label{eq:k2}	
	Pr[\textup{k  failures}]=\left(\frac{2^\lambdam-|\mathcal{P}(\lambdam)|}{2^\lambdam}\right)^k < \frac{1}{\lambdam^c}.
	\end{equation}
	Equation (\ref{eq:k1}) and (\ref{eq:k2}) give us the required bound on $k$:
	\begin{equation}
	\label{eq:k3}	
	\left(1-\frac{C_2}{\lambdam^2}\right)^k < \frac{1}{\lambdam^c} \Rightarrow
	k > -\frac{c\log(\lambdam)}{\log(1-C_2\lambdam^{-2})}.
	\end{equation}
	Based on logarithmic inequalities \cite{topsok2006some}, using the fact that 
	\begin{equation}
	-C_2\lambdam^{-2}>-1,\nonumber,
	\end{equation}
	we bound the denominator, 
	\begin{equation}
	\log(1-C_2\lambdam^{-2}) > \frac{-C_2\lambdam^
		{-2}}{1-C_2\lambdam^{-2}}.
	\label{eq:k4}
	\end{equation}
	We plug \eqref{eq:k4} into \eqref{eq:k3} and use simple algebraic operation to simplify the bound on~$k$:
	\begin{eqnarray}
	-\frac{c\log(\lambdam)}{\log(1-C_2\lambdam^{-2})} &<& -\frac{c\log(\lambdam)(1-C_2\lambdam^{-2})}{-C_2\lambdam^{-2}} \nonumber\\
	&=& \frac{c}{C_2}\log(\lambdam)(\lambdam^2-C_2) \nonumber\\
	&<& \frac{c}{C_2}(\lambdam^3-C_2\lambdam)\nonumber\\
	&\mysmaller& \ \frac{1}{C_2}(\lambdam^4-C_2\lambdam^2)  	
	\label{eq:k5}
	\end{eqnarray}
	
	From \eqref{eq:k5} and \eqref{eq:k2} it follows that  
	\begin{equation}
	\forall k\geq \frac{1}{C_2}(\lambdam^4-C_2\lambdam^2)\enspace : \enspace Pr[\textup{k  failures}] < \frac{1}{\lambdam^c}.\nonumber
	\end{equation}

	We showed that the probability that we couldn't find a safe prime after $k$ rounds, for $k > \frac{1}{C_2}(\lambdam^4-C_2\lambdam^2)$ is negligible. We need $\lambdam-1$ pseudo random bits in order to create a safe prime candidate per round. Meaning, the length of required pseudo random bits is polynomial in $\lambdam$.
\end{proof}

\section{\bname Commit Protocol} \label{appsec:bc commit protocol}
The difference of \bname from \alg was described in \S\ref{sec:bc:protocol}.
We detail the full \bname's commit protocol in Figure \ref{bc commit protocol}, highlighting the differences from \alg's commit protocol (Figure \ref{commit protocol}).
In table \ref{table:differences}, we describe each difference followed by line number in \bname's protocol and by the matching line number in \alg's protocol, if exists.
\begin{table*}
	\centering
	\begin{tabular}{ | p{9cm} | p{2.25cm} | p{2.25cm} | }
		\hline
		Description & \bname & \alg  \\ \hline
		The user creates a nonce by hashing the committed message. & \refp{line:bc gen nonce} & \refp{line:gen nonce} \\ \hline
		The first element in the seed Merkle tree is \textit{metadata}. & \refc{line:bc insert metadata} & -  \\ \hline
		The miner stores the order of users' nonces in a list $L_\textit{seed}$. & \refc{line:bc lseed1}, \refc{line:bc lseed2}, \refc{line:bc lseed3} & -  \\ \hline
		The parties compute the puzzle using \pcrdlalg. & \refp{line:bc gen puzzle} & \refp{line:gen puzzle}  \\
		\hline
		
	\end{tabular}
	\caption{\bname differences from \alg}
	\label{table:differences}
\end{table*}

\begin{figure*}[t]
	\footnotesize
	\begin{tikzpicture}[auto,>=latex']
	\begin{scope}
	\node[fit={(0,0) (0.25,0.25)}, inner sep=0pt,xshift=6.25cm, yshift=1.35cm] (Pi) {
		$\mathbf{P_i}$
	};
	\node[fit={(0,0) (0.25,0.25)}, inner sep=0pt,xshift=6.25cm, yshift=-5.25cm] (Pie) {
	};
	
	\node[fit={(0,0) (1.6,0.25)}, inner sep=0pt, xshift=8.85cm, yshift=1.35cm] (Pd) {
		\textbf{Coordinator}
	};
	
	\node[fit={(0,0) (1.6,0.25)}, inner sep=0pt,xshift=8.85cm, yshift=-5.25cm] (Pde) {};
	
	\algrenewcommand{\alglinenumber}[1]{\footnotesize p{#1}:}
	\renewcommand{\Statex}{\item[\hphantom{\footnotesize p\arabic{ALG@line}:}]}
	
	\node[fit={(0,0) (6,1)}, inner sep=0pt,yshift=0.5cm] (I1) {
		\begin{algorithmic}[1]
		\State $\textit{nonce}_i \gets$ \hl{$\fun{hash}(m_i)$}
		\label{line:bc gen nonce}
		\algstore{I}
		\end{algorithmic}
	};

	\node[fit={(0,0) (6,2.5)}, inner sep=0pt,yshift=-2.35cm] (I2b) {
		\begin{algorithmic}[1]
		\algrestore{I}
		\Statex \textcolor{gray}{Encrypt and Commit}
		\State {\label{line:verify_mp} assert(verifyMP(${seed}, \textit{MP}_i, \textit{nonce}_i$)})
		\State $Pk_\textit{TL} \gets$  \hl{$\fun{PCRGenPuz}$}($\lambdam$, seed)
		\label{line:bc gen puzzle}
		\State $r_i \randselect \{0,1\}^{\lambdah}$
		\label{line:bc gen_r}
		\State $C_{m_i} \gets \fun{enc}_{Pk_\textit{TL}}(m_i \| r_i)$
		\label{line:bc enc m}
		\State $H_{C_{m_i}} \gets \fun{hash}(C_{m_i})$
		\label{line:bc query random oracle}
		\algstore{I}
		\end{algorithmic}
	};
	\node[fit={(0,0) (6,1.5)}, inner sep=0pt, yshift=-4.25cm] (I3) {
		\begin{algorithmic}[1]
		\algrestore{I}
		\State assert({$H_{C_{m_i}} \text{ appears once in } L_H$})
		\label{line:bc verify appears once}
		\State $S_{{L_H},i}=\fun{sign}_{Sk_i}(L_H \parallel \textit{seed})$
		\label{line:bc sign_list}
		\end{algorithmic}
	};
	\algrenewcommand{\alglinenumber}[1]{\footnotesize c{#1}:}
	\renewcommand{\Statex}{\item[\hphantom{\footnotesize c\arabic{ALG@line}:}]}
	
	\node[fit={(0,0) (5,2)}, inner sep=0pt,yshift=-0.9cm, xshift=10.75cm] (D1) {
		\begin{algorithmic}[1]
		\Statex \textcolor{gray}{Generate Seed}
		\State Initiate Merkle tree $T_\textit{seed}$
		\State \hl{Insert $\fun{hash}(\textit{metadata})$ to $T_\textit{seed}$}
		\label{line:bc insert metadata}
		\Statex [Wait for $N$ nonces]
		\State \hl{$S_\textit{seed} \gets S_\textit{users} \cup i$}
		\label{line:bc lseed1}
		\State \hl{$L_\textit{seed} \gets \fun{PrepareList}(S_\textit{seed})$}
		\label{line:bc lseed2}
		\State Aggregate all nonces in  $T_\textit{seed}$ \hl{according to the order in $L_\textit{seed}$}
		\label{line:bc gen merkle tree}
		\label{line:bc lseed3}
		\State $\textit{MP}_i \gets$ Merkle proof that $\textit{nonce}_i$ is in $T_\textit{seed}$
		\label{line:bc nonce merkle proof}
		\State $\textit{seed} \gets T_\textit{seed}$'s root
		\label{line:bc seed is ready}
		\algstore{D}
		\end{algorithmic}
	};

	\node[fit={(0,0) (7,1.5)}, inner sep=0pt, yshift=-3.2cm, xshift=10.75cm] (D2) {
		\begin{algorithmic}[1]
		\algrestore{D}
		\Statex \textcolor{gray}{Form a Commitment List}
		\State $S_H \gets S_H \cup (H_{C_{m_i}}, i)$
		\label{line:bc gather Hc}
		\Statex [Wait until done for all parties]
		\State $L_H \gets \fun{PrepareList}(S_H)$
		\label{line:bc gen_list}
		\algstore{D}
		\end{algorithmic}
	};

	\node[fit={(0,0) (7,2.85)}, inner sep=0pt, yshift=-5.65cm, xshift=10.75cm] (D3) {
		\begin{algorithmic}[1]
		\algrestore{D}
		\Statex \textcolor{gray}{Collect Ciphertexts and Signatures}
		\State assert({\fun{verifySig}($Pk_i$, $L_H \parallel \textit{seed}$, $S_{{L_H},i}$)})
		\label{line:bc verify p sig}
		\State $j \gets $ get index of participant $i$ in $L_H$
		\State assert({$L_H[j] = (\fun{hash}(C_{m_i}),i)$})
		\label{line:bc verify p cipher}
		\State $L_s[j] \gets S_{{L_H},i}$
		\label{line:bc collect sigs}
		\State $L_C[j] \gets (C_{m_i}, i)$
		\label{line:bc collect C}
		\end{algorithmic}
	};
	
	\node[fit={(0,0) (0.25,0.25)}, inner sep=0pt,xshift=6.25cm, yshift=0.75cm] (I1p1) {};
	
	\node[fit={(0,0) (0.25,0.25)}, inner sep=0pt,xshift=9.5cm, yshift=-0.25cm] (D1p1) {};
	\node[fit={(0,0) (0.25,0.25)}, inner sep=0pt,xshift=9.5cm, yshift=-0.5cm] (D1p2) {};
	
	\node[fit={(0,0) (0.25,0.25)}, inner sep=0pt,xshift=6.25cm, yshift=-1.35cm] (I2p1) {};
	\node[fit={(0,0) (0.25,0.25)}, inner sep=0pt,xshift=6.25cm, yshift=-1.6cm] (I2p2) {};
	
	\node[fit={(0,0) (0.25,0.25)}, inner sep=0pt,xshift=9.5cm, yshift=-2.45cm] (D2p1) {};
	\node[fit={(0,0) (0.25,0.25)}, inner sep=0pt,xshift=9.5cm, yshift=-2.7cm] (D2p2) {};
	
	\node[fit={(0,0) (0.25,0.25)}, inner sep=0pt,xshift=6.25cm, yshift=-3.55cm] (I3p1) {};
	\node[fit={(0,0) (0.25,0.25)}, inner sep=0pt,xshift=6.25cm, yshift=-3.8cm] (I3p2) {};
	
	\node[fit={(0,0) (0.25,0.25)}, inner sep=0pt,xshift=9.5cm, yshift=-4.65cm] (D3p1) {};
	
	\path[-] (Pi) edge node {} (Pie);       
	\path[-] (Pd) edge node {} (Pde);
	\path[->] (I1p1) edge node[above, midway, sloped] {$(\textit{nonce}_i, i)$} (D1p1);
	\path[->] (D1p2) edge node[above, midway, sloped] {(seed,$\textit{MP}_i$)} (I2p1);
	\path[->] (I2p2) edge node[above, midway, sloped] {($H_{C_{m_i}}, i$)} (D2p1);
	\path[->] (D2p2) edge node[above, midway, sloped] {$L_H$} (I3p1);
	\path[->] (I3p2) edge node[above, midway, sloped] {$(S_{{L_H},i}, C_{m_i}, i)$} (D3p1);
	
	\node[fit={(0,0) (0,0)}, inner sep=0pt,xshift=0cm, yshift=0.4cm] (I1Ls) {};
	\node[fit={(0,0) (0,0)}, inner sep=0pt,xshift=6.25cm, yshift=0.4cm] (I1Le) {};
	\node[fit={(0,0) (0,0)}, inner sep=0pt,xshift=0cm, yshift=-2.9cm] (I2Ls) {};
	\node[fit={(0,0) (0,0)}, inner sep=0pt,xshift=6.25cm, yshift=-2.9cm] (I2Le) {};
	
	\path[-] (I1Ls) edge[dashed] node {} (I1Le);   
	\path[-] (I2Ls) edge[dashed] node {} (I2Le);   
	
	\node[fit={(0,0) (0,0)}, inner sep=0pt,xshift=9.75cm, yshift=-1.8cm] (D1Ls) {};
	\node[fit={(0,0) (0,0)}, inner sep=0pt,xshift=17cm, yshift=-1.8cm] (D1Le) {};
	\node[fit={(0,0) (0,0)}, inner sep=0pt,xshift=9.75cm, yshift=-3.25cm] (D2Ls) {};
	\node[fit={(0,0) (0,0)}, inner sep=0pt,xshift=17cm, yshift=-3.25cm] (D2Le) {};
	
	\path[-] (D1Ls) edge[dashed] node {} (D1Le);   
	\path[-] (D2Ls) edge[dashed] node {} (D2Le);   
	
	\draw [decorate,decoration={brace,amplitude=8pt},yshift=0pt]
	(10.03,-1.7) -- (10.03,1.35) node {};
	\draw [decorate,decoration={brace,amplitude=8pt},yshift=0pt]
	(10.03,-3.15) -- (10.03,-1.9) node {};
	\draw [decorate,decoration={brace,amplitude=8pt},yshift=0pt]
	(10.03,-5.25) -- (10.03,-3.35) node {};
	
	\draw [decorate,decoration={brace,amplitude=8pt},yshift=0pt]
	(6,0.25) -- (6,-2.8) node {};
	\draw [decorate,decoration={brace,amplitude=8pt},yshift=0pt]
	(6,-3) -- (6,-4.15) node {};
	
	\end{scope}
	\end{tikzpicture}
	\caption{\bname Commit Protocol}
	\label{bc commit protocol}
\end{figure*}
\end{document}